\documentclass[a4]{article}
\usepackage[affil-it]{authblk}

\newcommand\SANSCOMMENTAIRE[1]{#1}

\usepackage{url}
\providecommand\href[1]{\url{#1}}

\usepackage{microtype}

\usepackage{amsmath}
\usepackage{amssymb, amsfonts}
\usepackage[cyr]{aeguill}
\usepackage{url}

\usepackage[utf8]{inputenc}
\usepackage[T1]{fontenc}

\newcommand\vectorl[1]{{\mathbf#1}}

\newcommand\vx{\vectorl{x}}

\newcommand\R{\mathbb{R}}
\newcommand\Q{\mathbb{Q}}
\newcommand\N{\mathbb{N}}
\newcommand\Z{\mathbb{Z}}

\newcommand\tZ{\tilde{\Z}}

\newcommand\Greg{Grzegorczyk}
\newcommand\E{\mathcal{E}}
\newcommand\REC{{\rm REC}}

\newcommand\SMIN{{\rm SMIN}}
\newcommand\limREC{{\rm BR}}
\newcommand\BRN{{\rm BRN}}
\newcommand\BSUM{{\rm BSUM}}
\newcommand\BPROD{{\rm BPROD}}
\newcommand\BSUMs{{\rm BSUM_{<}}}
\newcommand\BPRODs{{\rm BPROD_{<}}}
\newcommand\ODE{{\rm ODE}}
\newcommand\limODE{{\rm boundedlinODE}}
\newcommand\LI{{\rm LI}}
\newcommand\Elem{\mathcal{E}}
\newcommand\Gregn{\mathcal{E}_n}
\newcommand\PR{\mathcal{P}\mathcal{R}}

\newtheorem{definition}{Definition}
\newtheorem{remark}{Remark}
\newtheorem{lemma}{Lemma}
\newtheorem{theorem}{Theorem}
\newtheorem{corollary}{Corollary}

\newtheorem{example}{Example}
\newenvironment{proof}{\noindent {\bf Proof}:}{\hfill $\Box$ \medskip}
\newenvironment{proofof}[1]{\noindent {\bf Proof (#1)}:}{\hfill
$\Box$ \medskip}

\def\propositions{theorem}

\newenvironment{oureqnarray}{\begin{equation}\begin{array}{lll}}{\end{array}\end{equation}}

\newcommand\dint[4]{\int_{#1}^{#2}{#3}{\delta #4}}

\newcommand\onehalf{\frac{1}{2}}

\newcommand\fallingexp[1]{\overline{2}^{#1}}

\usepackage{color}
\usepackage{todonotes}

\newcommand{\arnaud}[2][]{\SANSCOMMENTAIRE{\todo[inline,color=green!20,caption={2do}, #1]{\begin{minipage}{\textwidth-4pt}
Arnaud:			#2\end{minipage}}}}

\newcommand{\arnaudmargin}[1]{\SANSCOMMENTAIRE{\todo[color=green!40]{A: #1}}}
\newcommand{\olivier}[2][]{\SANSCOMMENTAIRE{\todo[inline,color=red!40,caption={2do}, #1]{\begin{minipage}{\textwidth-4pt}
Olivier:			#2\end{minipage}}}}

\newcommand{\sabrina}[2][]{\SANSCOMMENTAIRE{\todo[inline,color=blue!20!white,caption={2do},#1]{\begin{minipage}{\textwidth-4pt}
Sabrina:			#2\end{minipage}}}}

\newcommand{\fonction}[1]{\textsf{#1}}

\newcommand\projection[2]{\mathbf{\pi}_{#1}^{#2}}
\newcommand{\sucs}{\mathbf{s}}
\newcommand\myominus{\mathbf{\ominus}}
\newcommand\plus{\mathbf{+}}
\newcommand\minus{\mathbf{-}}
\newcommand\gE{\mathbf{E}}

\newcommand{\succun}[1]{\mathbf{1}({#1})}
\newcommand{\succzero}[1]{\mathbf{0}({#1})}

\def\moins{\mathrel{\dot{-}}}

\newcommand{\zero}{\mathbf{0}}

\newcommand{\sign}[1]{\fonction{sg}(#1)}
\newcommand{\signn}[1]{\fonction{sg}_{\N}(#1)}
\newcommand{\signcomp}[1]{\bar{\fonction{sg}}(#1)}
\newcommand{\signcompn}[1]{\bar{\fonction{sg}_\N}(#1)}

\newcommand{\cond}[3]{\fonction{if}(#1,#2,#3)}
\newcommand{\condn}[3]{\fonction{if}_\N(#1,#2,#3)}

\newcommand{\case}{\fonction{case}}

\newcommand{\tu}[1]{\mathbf{#1}}

\newcommand{\dderiv}[2]{\frac{\partial #1}{\partial #2}}
\newcommand{\dderivL}[1]{\frac{\partial #1}{\partial \lengt}}
\newcommand{\dderivl}[1]{\frac{\partial #1}{\partial \ell}}

\newcommand{\lengt}{\mathcal{L}}
 \newcommand\lengthnotation{\ell}
\newcommand{\length}[1]{\mathrm{\lengthnotation}(#1)}
\newcommand{\degre}[1]{\mathrm{deg}(#1)}

\newcommand{\inst}{\mathsf{inst}}

\newcommand{\derivlength}{\mathbb{DL}}
\newcommand{\sll}{\mathbb{SLL}}

\newcommand{\nextI}{\mathsf{next}}

\newcommand{\cp}[1]{\mathbf{#1}}
\newcommand{\Ptime}{\cp{PTIME}}      
\newcommand{\NPtime}{\cp{NPTIME}}
\newcommand{\FPtime}{\cp{FPTIME}}
\newcommand{\cPspace}{\cp{\sharp PSPACE}}
\newcommand{\Pspace}{\cp{PSPACE}}
\newcommand{\FPspace}{\cp{FPSPACE}}
\newcommand{\FPspaceN}{\cp{FPSPACE}}

\newcommand{\myFPspace}{\mathcal{F}_\cp{PSPACE}}

\newcommand{\suffix}{\textsf{suffix}}

\begin{document}


\newcommand\myaddress{
{\sc LIX, Ecole Polytechnique,}\\ 91128 Palaiseau Cedex, FRANCE \\
{\small Olivier.Bournez@lix.polytechnique.fr}}

\title{Recursion schemes,  discrete  differential equations and characterization of polynomial time computation}
\date{\today}

\author{Olivier Bournez \thanks{Email: bournez@lix.polytechnique.fr. Supported  by
  RACAF Project from Agence National de la
    Recherche  and Labex Digicosme Project ACDC.}}\affil {Ecole Polytechnique, LIX, 91128 Palaiseau
    Cedex, France}

 \author{Arnaud Durand \thanks{Email: durand@math.univ-paris-diderot.fr}}\affil{Université Paris Diderot, IMJ-PRG, CNRS
   UMR 7586, Case 7012, 75205 Paris cedex 13, France}

 \author{Sabrina Ouazzani \thanks{Email: sabrina@lix.polytechnique.fr. Supported by Labex Digicosme Project ACDC.}}\affil{LIX, 91128 Palaiseau
   Cedex, France  and LACL, Universit\'e Paris-Est Cr\'eteil, 61 avenue du
   G\'en\'eral de Gaulle, 94010 Cr\'eteil,
   France}


\maketitle

\begin{abstract}
	
This papers studies the expressive and computational power of 	 discrete Ordinary Differential Equations (ODEs).
It presents a  new framework using discrete ODEs as a central tool for
computation 
and provides several 
implicit characterizations of complexity and computability
classes. 

The proposed framework presents an original point of view on
complexity and computability classes. It also unifies in an
elegant settings various constructions that have been proposed for characterizing these classes. This includes Cobham's
and, Bellantoni and Cook's definition of
polynomial time and later extensions on the approach, as well as
recent characterizations of computability and complexity by classes of
ordinary differential equations.
It also helps understanding the relationships between analog
computations and classical discrete models of computation theory.


At a more technical point of view, this paper points out the fundamental role of linear (discrete)
ordinary differential equations and classical ODE tools such as changes of variables to capture computability and complexity
measures, or as a tool for programming various algorithms.


\end{abstract}

\tableofcontents


%
%

\section{Introduction}

Since the beginning of its foundations, classification of problems, by
various models of computation, 
either by their complexity or by their computability properties, is a
thriving field of computer science. 
Nowadays, classical (digital) computer science problems
also deal with continuous data coming from different areas and modeling
involves the use of tools like numerical analysis, probability theory
or differential equations. Thus new characterizations related to
theses fields have been proposed.
On a dual way,  the quest for
new type of computers recently leaded to revisit the power of some
models for analog machines based on differential equations, and to
compare them to modern digital models.
In both contexts, when discussing the related issues, one has to
overcome the fact that today's (digital) computers are by essence
discrete machines while the objects under study are continuous and
naturally correspond to Ordinary Differential Equations (ODEs).

We consider here an original approach in
between the two worlds: discrete oriented computation with
differential equations.
 
Indeed, ODE 
appear to be a natural way of expressing properties and
are intensively used, in particular in applied science. 
 The theory of classical (continuous) ODEs is rather very well
understood, broadly studied and taught with a plethoric literature, see
e.g. \cite{Arn78, BR89, CL55}.
%
We are here interested here in a discrete more recent
counterpart of classical continuous ODEs: discrete ODEs, whose
theory is sometimes called \emph{discrete calculus}, or \emph{finite
  calculus}.  See for e.g.
\cite{graham1989concrete,gleich2005finite,izadi2009discrete,urldiscretecalculuslau}
or our brief presentation in Appendix \ref{sec:dcal}.

In this article, we prove that various classes of complexity and
computability theory can be very elegantly and simply defined using
discrete ordinary differential equations.
We also demonstrate through this discussion how some techniques from
analog world such as changes of variables can be used to solve
efficiently some (classical, digital) problems.

As far as we know, this is the first time that computations with
discrete ODEs and their related complexity aspects have been considered. By contrast, characterizations have been recently obtained
with classical (continuous) ODEs of various classes of functions,
mostly in the framework of computable analysis. The hardness of
solving continuous ODEs has been intensively discussed: for example
\cite{Ko83} establishes somes basis of the complexity aspects of ODEs
and more recent work like \cite{kawamura2009lipschitz} or
\cite{collins2008effectivesimpl} establishes links between complexity
or effective aspects of such differential equations.
Hence, the computational power of continuous ODEs is much more understood. 

We believe that investigating the expressive power of discrete ODE, can really help to better understand complexity of computation for both the discrete and continuous settings. 
Indeed, on one hand, as a consequence, our work relates classical
(discrete) complexity classes to analog computations,
i.e. computations over the reals, as analog computation have been
related in various ways to continuous ordinary differential equations,
and as discrete ordinary differential equations provide clear hints
about their continuous counterparts.
But on the other hand, it also opens a new perspective on classical
discrete computations, i.e. computation that deals
with bits, words, or integers. In this discrete setting,
our work falls under the scope of so-called implicit complexity, i.e.
characterization complexity measures in a machine independent
way. Combining these two approaches, it helps to clearly point out
which aspects of the statements are related to continuous computations
versus discrete computations.

This original work point out the fundamental role of linear (discrete)
ordinary differential equations in computability and complexity
theory: when considered in general, this provides a characterization
of elementary functions. When considered with suitable (length
related) changes of
variables, this provides a characterization of polynomial time. This
work also opens a way to revisit seminal results such as Cobham's  \cite{cob65} and Bellantoni and Cook's  \cite{bs:impl} definition of
polynomial time as syntactic constraints imposing only linear discrete
ODEs. 



\subsection{Related works}

%
\subsubsection{Analog computations:} 
In the context of analog computations there have been several results
relating classical complexity to various classes of continuous
ODEs. In particular, a serie of papers has been devoted to study
various classes of the so-called $\R$-recursive functions, after
their introduction in \cite{Moo95b}.  At the complexity level,
characterizations of complexity classes such as $\Ptime$ and $\NPtime$ using
$\R$-recursive algebra have been obtained \cite{MC06}, motivated in
particular by the idea of transferring classical questions from
complexity theory to the context of real and complex analysis
\cite{LCM07b, MC06,MC05Eatcs}. More recently, is has been proved that
polynomial differential equations can be considered as a very simple
and elegant model in which computable functions over the reals and
polynomial time computable functions over the reals can be defined
without any reference to concepts from discrete computation theory
\cite{JournalACM2017,TheseAmaury}.  
Refer to
\cite{DBLP:journals/corr/BournezGP16,ArxivSurvey} for an up to date
survey about various works on analog computations, in particular in a
computation theory perspective. 

%
%
%

\subsubsection{Classical complexity theory:}  
Implicit complexity has been developed in many ways
to provide machine independant characterizations of various computability and complexity
classes in the discrete setting. 
This includes 
characterizations of complexity classes based
on lambda calculus (e.g.\ \cite{LM93}), finite model theory and descriptive complexity (e.g.\
\cite{Fag74}), on function algebra (e.g.\ \cite{cob65,bs:impl,lm:pspace,Rose}), or yet one
combining the latter two approaches (e.g.\ \cite{Gur83,Saz80}). 
This approach has also been proved useful to measure the expressive
power of various formalisms with wide applications in database and
constraint theory and programming languages: See \cite{Clo95,LivreFiniteModelTheory,Imm99} for more complete
references.


\subsection{Structure of the paper}
In Section~\ref{sec:discrete differentiability} a short introduction to discrete differentiability is given followed in Section~\ref{sec:programming with ODE} by an illustration, through examples, of the programming ability of discrete ODE.
First formal definitions of discrete ODE are given in
Section~\ref{sec:Computability and Discrete ODEs} together with
characterizations of primitive recursion and elementary functions in
this context. Section~\ref{sec:restrict} introduces the notion (and
basic theory) of length-ODE which is central in the characterization
of $\FPtime$ (Section~\ref{sec:A characterization of polynomial
  time}). Section~\ref{sec:extension}
discusses some extensions of the results. Details about discrete
differentiability and the associated calculus, as well about Random
access machines are given in
Appendix. 

\section{Discrete differentiability}\label{sec:discrete differentiability}


In this section, we review some basic notions of discrete calculus to help intuition in the rest of the paper (see Appendix~\ref{sec:dcal} for a more complete review).

\subsection{Discrete derivation, integration and exponentiation}

Discrete derivatives are usually intended to concern functions over the
integers of type $\tu f: \N^p \to \Z^d$, but the statements and concepts
considered in our discussions  are also valid more generally for functions
of type $\tu f: \Z^p \to \Z^d$, for some integers $p,d$, or even functions
$\tu f: \R^p \to \R^d$.
The basic idea is to consider the following concept of
derivative.

\begin{definition}[Discrete Derivative] The discrete derivative of
	$\tu f(x)$ is defined as $\Delta \tu f(x)= \tu f(x+1)-\tu
        f(x)$. We will also write $\tu f^\prime$ for
        $\Delta \tu f(x)$ to help to understand
	statements with respect to their classical continuous counterparts. 
	
\end{definition}


Several results from classical derivatives generalize
to these settings: this includes linearity of derivation $(a \cdot
f(x)+ b \cdot g(x))^\prime = a \cdot f^\prime(x) + b \cdot
g^\prime(x)$, formulas for products 
and division such as
 $(f(x)\cdot g(x))^\prime =
 f^\prime(x)\cdot g(x+1)+f(x) \cdot g^\prime(x)$. 
A fundamental concept is the following:

\begin{definition}[Discrete Integral]
	Given some function $\tu f(x)$, we write $\dint{a}{b}{\tu f(x)}{x}$
	as a synonym for $$\dint{a}{b}{\tu f(x)}{x}=\sum_{x=a}^{x=b-1}
        \tu f(x)$$ with the convention that it values $0$ when $a=b$ and
        $\dint{a}{b}{\tu f(x)}{x}=- \dint{b}{a}{\tu f(x)}{x}$ when $a>b$. 
\end{definition}

The telescope formula yields the so-called Fundamental Theorem of
Finite Calculus: 

\begin{theorem}[Fundamental Theorem of Finite Calculus]
	Let $\tu F(x)$ be some function.
	Then,
	$$\dint{a}{b}{\tu F^\prime(x)}{x}= \tu F(b)-\tu F(a).$$
      \end{theorem}

As for classical functions, any function has a set of primitives
defined up to some additive constant, and techniques such as
integration by part can be used.

 \newcommand\ENONCEINTEGRALPARAMETER[1]{
   
 \begin{lemma}[Derivation of an integral with parameters] #1
   Consider $$\tu F(x) = \dint{a(x)}{b(x)} {\tu f(x,t)}{t}.$$
 Then $$\tu F'(x) = \dint{a(x)}{b(x)}{  \frac{\partial \tu f}{\partial
     x} (x,t)}{t} + \dint{0}{-a^\prime(x)}{\tu f(x+1,a(x+1)+t)}{t}
+ \dint{0}{b'(x)}{ \tu f(x+1,b(x)+t ) } {t} 
$$

In particular, when $a(x)=a$ and $b(x)=b$ are constant functions, $$\tu F'(x) = \dint{a}{b}{  \frac{\partial \tu f}{\partial
     x} (x,t)}{t}$$ 
\end{lemma}
}

{
\ENONCEINTEGRALPARAMETER{}

}

A classical concept in discrete calculus is the one of falling
	power defined as $x^{\underline{m}}=x\cdot (x-1)\cdot (x-2)\cdots(x-(m-1))$.
	This notion is  motivated by a derivative formula similar to the classical
one for powers in the continuous setting.
In a similar spirit, we 
introduce the following concept. This seems not standard (as far as
the authors know) but of clear interest.


\begin{definition}[Falling exponential]
	Given some function $\tu U(x)$, the expression $\tu U$ to the
	falling exponential $x$,
	denoted by $\fallingexp{\tu U(x)}$, stands
        for  $$\fallingexp{\tu U(x)}= (1+ \tu U^\prime(x-1)) \cdots
        (1+ \tu U^\prime(1)) \cdot (1+ \tu U^\prime(0))  = 
	\prod_{t=0}^{t=x-1} (1+ \tu U^\prime(t)). $$
	with the convention that $\prod_{0}^{0}=\tu {id}$, where $\tu
        {id}$ is the identity (sometimes denoted  $1$ hereafter)
\end{definition}

This is motivated by the remark that for all $x \in \Z$,
	$2^x=\fallingexp{x}$, and:

\begin{theorem}[Derivative of a falling exponential] The discrete
	derivative of a falling exponential is given by
	$$\left(\fallingexp{\tu U(x)}\right )^\prime = \tu U^\prime(x) \cdot
	\fallingexp{\tu U(x)}.$$
	
\end{theorem}

\subsection{Discrete Ordinary Differential Equations (ODE)}

We will focus in this article on discrete Ordinary Differential
Equations (ODE) on functions with several variables, that is to say
for example on equations of the (possibly vectorial) form:
%
%
%
\begin{equation}\label{lodepv}
\dderiv{\tu f(x,\tu y)}{x}= \tu h(\tu f(x,\tu y),x,\tu y).
\end{equation}

As expected $\dderiv{\tu f(x,\tu y)}{x}$ stands for the derivative of
functions $f(x,\tu y)$ considered as a function of $x$, when $\tu y$
is fixed.
When some initial value is given, this is called an \emph{Initial Value
  Problem (IVP)}, also called a \emph{Cauchy Problem}.
{That is to say, we are
given a problem of type:

\begin{oureqnarray}\label{cauchy}
\dderiv{\tu f(x,\tu y)}{x}&=&\tu h(\tu f (x,\tu y),x,\tu y )\\
\tu f(0,\tu y)&=& \tu g(\tu y )
\end{oureqnarray}
\noindent with functions $\tu g,\tu h$ of suitable dimensions and domains:
}
Our aim here is to discuss total functions whose domain and range
is either of the form $\mathcal{D}=\N$, $\Z$, or possibly a finite
product $\mathcal{D}= \mathcal{D}_1 \times \dots \times \mathcal{D}_k$
where each $\mathcal{D}_i=\N$, $\Z$.
By considering that $\N \subset \Z$, we  assume that the range is
always 
$\Z^d$ for some $d$. 
The concept of solution for such ODEs is  as expected: assume $h: \Z^d \times \N \times \Z^{p} \to \Z$
(or $h: \Z^d \times \Z \times \Z^{p} \to \Z$),
a solution  
over $\mathcal{D}$ is a function $f: \mathcal{D}
\times \Z^{p} \to \Z^d$
that satisfies the equations for all $x,\tu y$.

We will only consider well-defined ODEs such as above in this
article (but variants with partially defined function could be considered as well). 
Observe that an IVP of the form \eqref{lodepv} always admits a
(necessarily unique) solution over $\N$ since $f$ can be defined
inductively with $\tu f(0,\tu y)=\tu g(\tu y)$ and $\tu f(x+1,\tu
y)= \tu f(x,\tu y)
+ \tu h(\tu f(x,\tu y),x, \tu y)$. 

\begin{remark}
Notice that this is not necessarily true over $\Z$: As an
example, consider 
$f^\prime(x) = - f(x) + 1$, $f(0) = 0$. 
By definition of $f^\prime(x)$, we must have $f(x+1)=1$ for all $x$,
but if $x = -1$, $f(0)=1 \neq 0$. 

\end{remark}

\begin{remark}[Sign function] It is very instructive to realize that the solution of this
  IVP 
  over $\N$ is the sign $\signn{x}$ function defined by
  $\signn{x}=1$ if $x>0$ and $\signn{x}=0$ in the other case.
\end{remark}


Affine (also called linear) ordinary differential equations will play a very important role
in what follows, i.e. discrete ordinary differential equations of the
form $\tu f^\prime(x)=\tu A(x) \cdot \tu f(x) + \tu B(x)$.

\begin{remark}
Recall that the solution of $ f^\prime(x)= a(x) f(x) + b(x)$ for
classical continous derivatives turns out to be given by (This is usually obtained by the variation of parameter method
     - see Appendix \ref{sec:dcal} for a short review of the method
     and for other classical ODE): 
	\begin{equation} \label{eq:affine}
	f(x) = f(0) e^{\int_0^{x} a(t)
		dt}  + \int_{0}^{x} b(u) e^{\int_u^{x} a(t)
		dt }
	du.
	\end{equation}

            This 
        generalizes to discrete ordinary
differential equations, and this works even vectorially: 
        \end{remark}


\begin{lemma}[Solution of ODE $\tu f’(x,\tu y)= \tu A (x,\tu y) \cdot
	\tu f(x,\tu y)
	+ \tu B(x,\tu y)$] \label{def:solutionexplicitedeuxvariables}
	For matrices $\tu A$ and vectors $\tu B$ and $\tu G$,
	the solution of equation $\tu f’(x,\tu y)= \tu A(x,\tu y) \cdot \tu f(x,\tu y)
	+  \tu
	B (x,\tu y)$  with initial conditions $\tu f(0,\tu y)= \tu G(\tu y)$ is
	\begin{equation}\label{soluce}
          \tu f(x,\tu y) = \left( \fallingexp{\dint{0}{x}{\tu A(t,\tu y)}{t}} \right) \cdot \tu G (\tu y) +
	\dint{0}{x}{ \left(
		\fallingexp{\dint{u+1}{x}{\tu A(t,\tu y)}{t}} \right) \cdot
              \tu B(u,\tu y)} {u}.
          \end{equation}

        \end{lemma}

\section{Programming with discrete ODE}\label{sec:programming with ODE}

In this section, we show that several algorithms can actually be
naturally solved using discrete ODEs, or viewed as discrete ODEs:
Basically, for now, we suppose that composition of functions, constant
and  the following basic functions can be used freely as functions
from $\Z$ to $\Z$:
%
%
\begin{itemize}
	\item arithmetic operations: $+$, $-$, $\times$
        \item $\length{x}$ returns the length in binary of $x \in \N$. 
	\item $\sign{x}: \Z \to \Z$ defined by $\sign{x}=1$ if $x>0$ and
          $\sign{x}=0$ in the other case. 
  $\signn{x}: \N \to \Z$ defined by $\signn{x}=1$ if $x>0$ and
          $\signn{x}=0$ in the other case.   
\end{itemize}
Recall that $\signn{x}$ is
     the solution over $\N$ of  some IVP 
          and hence is very natural in this context.
%
From these basic functions, for readability, one may define useful functions as
synonyms: 
\begin{itemize}
	\item $\signcomp{x}$ stands for 
          $\signcomp{x}=(1-\sign{x})\times (1-\sign{-x})$: it tests if
          $x=0$ for $x \in \Z$.
  $\signcompn{x}$ stands for 
          $\signcompn{x}=1-\signn{x}$: it tests if $x=0$ for $x \in \N$. 
\item $\cond{x}{y}{z}$ stands for $\cond{x}{y}{z}=y +
  \signcomp{x}\cdot (z-y)  $ and 
$\condn{x}{y}{z}$ stands for $\condn{x}{y}{z}=y +
  \signcompn{x}\cdot (z-y)  $.
We have for both versions (The point is that the first considers $x \in \Z$ while the second
assumes $x \in \N$): 
\[
\cond{x}{y}{z}=\left\{\begin{array}{l}
y \mbox{ if } x=0\\
z \mbox{ otherwise }
\end{array}\right.
\quad
\condn{x}{y}{z}=\left\{\begin{array}{l}
y \mbox{ if } x=0\\
z \mbox{ otherwise }
\end{array}\right.
\]

\item 
We will
extensively use these functions below: 
  $\cond{x<x'}{y}{z}$ will be a synonym for
  $\cond{\sign{x-x'-1}}{y}{z}$.
  $\cond{x=x'}{y}{z}$
  will be a synonym for $\cond{\signcomp{x-x'}}{y}{z}$.
\end{itemize}

\begin{example}[Computing the minimum of a function: ]
The minimum of a function 
$\min f : x \mapsto min \{f(y) : 0 \le y \le x\}$
is given by $F(x,x)$ where $F$ 
 can be computed recursively by
\begin{eqnarray*}
F(0) &=& f(0) \\
F(t+1,x) &=& \textsf{if}(F(t,x) < f(x), F(t,x), f(x)) 
\end{eqnarray*}
%
%
%
This can be interpreted as a discrete ordinary differential equation:

$$\dderiv{F(t,x)}{t} =  H(F(t,x),f(x),t,x)=
\left\{
\begin{array}{lll}
0  & \mbox{ if } & F(t,x)<f(x) \\
- F(t,x) + f(x)  & \mbox{ if } &  F(t,x) \ge f(x).\\
\end{array}
\right.
$$

The value $\min f (x)  = F(x,x)$ can then be computed 
using the solution of the above discrete ODE. In integral form, we
have: 
\[
F(x,y) = F(0) + \dint{0}{x}{H(F(t,y),t,y)}{t}.
\]
\end{example}

\begin{remark}
We also see through this example that such an integral (equivalently
discrete ODE) can 
always be considered as a (recursive) algorithm:
compute the integral from its definition as a sum to compute the
function.
Notice that this algorithm is not polynomial 
as this basically takes
time $x$ to compute $\min f$,
i.e. not polynomial with respect to the usual convention for
measuring complexity based on the  binary length of arguments. 
\end{remark}

\begin{example}[Computing the integer part and divisions, 
  Length-ODE:] \label{ex:some}
\label{sec:racine}
Suppose  that we want to compute $\lfloor \sqrt{x} \rfloor = \max \{ y \le x : y \cdot y \leq x \}$ and $\lfloor
\frac{x}{y} \rfloor = \max
\{z \le x : z \cdot y = x \}$. It can be done by the following uniform method. Let $f,h$ be some functions with $h$ being non decreasing. We compute $\fonction{some}_h$  with $\fonction{some}_h(x)=y$ s.t. $|f(x)-h(y)|$ is minimal.   
When $h(x)=x^2$ and $f(x)=x$, it holds that:
$
\lfloor \sqrt{x} \rfloor = \textsf{if}(\fonction{some}_h(x)^2 \leq x,\fonction{some}_h(x),\fonction{some}_h(x)-1).
$

The function $\fonction{some}_h$ can be computed 
as a solution of an ODE 
as in the preceding example. However, there is a more efficient way to do it based on what one 
usually does 
with classical ordinary differential equations: 
performing a
change of variable so that it becomes logarithmic in $x$.
Indeed, we write
$
\fonction{some}_h(x)  = G(\length{x},x)
$ 
for some function $G(t,x)$ defined by:
\begin{eqnarray*}
	G(0,x)&=&x \\
	G(t+1,x) &=& \textsf{if}( h(G(t,x)) = f(x), G(t,x) , \\
	&& \textsf{if}(h(G(t,x)) > f(x), G(t,x) + 2^{\length{x}-t-1}, G(t,x) -2^{\length{x} - t -1})) 
\end{eqnarray*}
Or, if one prefers,  $G(t,x)$ is solution of 
$$\dderiv{G(t,x)}{t} =  E(G(t,x),t,x)=
\left\{
\begin{array}{lll}
+ 2^{\length{x}-t-1} & \mbox{ if } & h(G(t,x))>g(x) \\
0 & \mbox{ if } &  h(G(t,x))=g(x)\\
- 2^{\length{x}-t-1} & \mbox{ if } & h(G(t,x))<g(x) \\
\end{array}
\right.
$$

This is indeed a differential equation whose solution is
converging fast (in polynomial time) to what we want. Reformulating  
what we just did, 
we wrote $\fonction{some}_h(x)  = G(\length{x},x)$
using the solution of the above discrete ODE, i.e. the solution of 
$
G(T,y) = x + \dint{0}{T}{E(G(t,y),t,y)}{t}.
$
This provides a polynomial time algorithm to solve our problems 
 using
a new parameter $t=\length{x}$ logarithmic in $x$. Such techniques
will be at the heart of the coming results.
\end{example}

Notice that the theory of ODEs also provides very natural alternative
ways to compute various quantities. This is very clear when
considering numeric functions such as $\tan$, $\sin$, etc. 

\begin{example}[Computing $\tan$ with discrete ODEs, iterative algorithms]




As an illustration, suppose you want to compute $\tan(x_0)$ for say $x_0=72$.
One way to do it is to observe that $\tan(x)^\prime= 	(1+\tan(x) \tan(x+1))$. 
From fundamental theorem of finite calculus 
we can hence write:
\begin{eqnarray}
\label{eq:tran}
  tan(x_0)&=&0 + \dint{0}{x_0}{\tan(x)}{x}  \label{eq:unt}
  \\
&=& 0 + \tan(1) \cdot 
  \dint{0}{x_0}{(1+\tan(x)\tan(x+1))}{x}  \label{eq:untt}
\end{eqnarray}

Inspired from previous remarks, the  point is that Equation
\eqref{eq:unt} can be interpreted as an algorithm: it provides a
way to compute $\tan(x_0)$ as an integral (or if you prefer as a
sum).

Thinking about what means this integral, discrete ODE  \eqref{eq:tan},
also encoded by \eqref{eq:untt}, can also be interpreted as
$\tan(x+1)-\tan(x)= \tan(1) \cdot  [1 + \tan(x) \tan(x+1)]$ 
  that is to say
 $\tan(x+1)= f(\tan(x))$ 
where   $f(X)= \frac{X+\tan(1)}{1-\tan(1)X}$. 
Hence, this is suggesting  a way to compute $\tan(72)$ by a method close to
express that $tan(x_0)=f^{[x_0]}(0).$ 
That is to say Equations \eqref{eq:unt} and \eqref{eq:untt} can be
interpreted as providing a way to compute $\tan(72)$ using an
iterative algorithm: they basically encode some recursive way of
computing $\tan$.
\end{example}












Of course, a similar principle would hold for $\sin$, or $\cos$ using
discrete ODEs obtained above, and for many other functions starting
from expression of their derivative. 

%
%

\begin{remark}
Given $x_0$, (even if we put aside how to deal with involved real
quantities) a point is that computing $\tan(x_0)$ using this method can not be
considered as
polynomial time, as the (usual) convention is that
time complexity is measured in term of the length of $x_0$, and not on
$x_0$.
\end{remark}

Could we do the same computation faster using a change of variables?
This is at the heart of the coming constructions and discussions.


\begin{example}[Computing suffixes with discrete ODEs]
Discrete ODEs turns out to be very natural in many other contexts, in particular non numerical ones,
where they would probably not be expected. 
We illustrate the discussion by a way to compute fast (in polynomial
time) the suffix function:
{
The suffix function, $\suffix(x,y)$ takes as input two integers $x$ and $y$ and outputs  the $\length{y}=t$ least significant bits of the binary decomposition of $x$.
We describe below a way to compute a suffix working over a
parameter $t$, that is logarithmic in $x$.}
Consider the following amazing algorithm that can be interpreted as a
fix-point definition of the function: $\suffix(x,y)=F(\length{x},y)$ where



 %
\[
F(T,y) = x + \dint{0}{T}{
\cond{\length{F(t,x)} = 1}{0}{- 2^{\length{F(t,x)} -1}}
}{t}.
\]


\end{example}

\section{Bounded schemes in computation
  theory}

After this teaser, the rest of this article aims at discussing which
problems can be solved using discrete ordinary differential equations,
and with which complexity.  Before doing so, we need to 
review some
basic concepts and results from computation theory that we will be
needed in the rest of this article and that have been obtained at this date.

\subsection{Computability theory and bounded schemes}

Classical recursion theory deals with functions over integers, that is
to say with functions $\tu f: \N^p \to \N^d$ for some positive integers
$p,d$. 

It is well known that all main classes of classical recursion theory can be
characterized as closures of a set of 
basic functions by a finite number of  basic
rules to build new functions: See
e.g. \cite{Rose,Odi92,clote2013boolean}: 

%

\begin{theorem}[Total Recursive functions]
  A total function over the integers is computable if and only if it belongs
  to the smallest set of functions that contains constant function
  $\zero$, the projection functions $\projection{i}{p}$, the function successor $\sucs$,
  that is closed under composition, primitive recursion and
  safe minimization.
\end{theorem}

In this statement, $\zero$, $\projection{i}{p}$ and $\sucs$ are
respectively the functions from, $\N \to \N$, 
$\N^p \to \N$ and $\N \to \N$ defined as  $n \mapsto 0$,
$(n_1,\dots,n_p) \mapsto n_i$, and $n \mapsto n+1$. 

We also recall here the basic definitions used in the above statement:

\begin{definition}[Primitive recursion] \label{def:classique}
Given functions $g: \N^p \to \N$ and $h: \N^{p+2} \to \N$, function
$f=\REC(g,h)$ defined by primitive recursion from $g$ and $h$ is the function $\N^{p+1} \to \N$ satisfying
\begin{eqnarray*}
f(0,\tu y) &=& g(\tu y) \\
f(x+1,\tu y)&=&h(f(x,\tu y),x,\tu y). \\
\end{eqnarray*}
\end{definition}

\begin{definition}[Primitive recursive functions]
A function over the integers is primitive recursive if and only if it belongs
  to the smallest set of functions that contains constant function
  $\zero$, the projection functions $\projection{i}{p}$, the functions successor $\sucs$,
  that is closed under composition and primitive recursion. 
\end{definition}

Primitive recursive functions have been stratified into various subclasses. 
We recall here the \Greg{} hierarchy in the rest of this subsection.

\begin{definition}[Bounded sum] \label{def:bsumclassique}
Given functions $g(\tu y): \N^p \to \N$,
\begin{itemize}
\item function $f=\BSUM(g): \N^{p+1} \to \N$ is defined by
  $f : (x,\tu y) \mapsto \sum_{z \le x} g(z,\tu y)$.
\item function
  $f=\BSUMs(g): \N^{p+1} \to \N$ is defined by
  $f : (x,\tu y) \mapsto \sum_{z < x} g(z,\tu y)$ for $x \neq 0$, and
  $0$ for $x=0$.
\end{itemize}

\end{definition}

\begin{definition}[Bounded product] \label{def:bprodclassique}
Given functions $g: \N^p \to \N$,
\begin{itemize}
\item function $f=\BPROD(g) : \N^{p+1} \to \N$ is defined by
$f: (x,\tu y) \mapsto \prod_{z \le x} g(z,\tu y)$.
\item function $f=\BPRODs(g)$ is defined by
$f: (x,\tu y) \mapsto \prod_{z < x} g(z,\tu y)$ for $x \neq 0$, and
$1$ for $x=0$.
\end{itemize}
\end{definition}

We have 
\begin{eqnarray*}
\BSUM(g)(x,\tu y) &=&\BSUMs(g)(x,\tu y)+g(x,\tu y) \\
\BPROD(g)(x,\tu y)&=&\BPRODs(g)(x,\tu y) \cdot g(x,\tu y)\\
\end{eqnarray*}


\begin{definition}[Elementary functions]
A function over the integers is elementary if and only if
it belongs to the smallest set of functions that contains constant
function $\zero$, the projection functions $\projection{i}{p}$, the functions
successor $\sucs$, addition $\plus$, limited subtraction $\myominus: (n_1,n_2) \mapsto
max(0,n_1-n_2)$, and that is closed under composition, bounded sum
$\BSUM$ and
bounded product $\BPROD$.

We denote by $\Elem$ the class of elementary functions.
\end{definition}

Class $\Elem$ contains many classical functions.
In particular:

\begin{lemma}[{\cite[Lemma 2.5, page 6]{Rose}}] $(x,y)\mapsto \lfloor
  x/y \rfloor$ is in $\Elem$. \label{lem:derose}
\end{lemma}
\begin{lemma}[{\cite{Rose}}] $(x,y) \mapsto x \cdot y$ is in $\Elem$.
\end{lemma}

The following normal form is also well-known. We consider safe
minimization instead of classical minimization as we focus in this
article only on total functions. 


%

\begin{definition}[(Safe) Minimization] \label{def:classiquesm}
Given function $g: \N^{p+1} \to \N$, such that for all $x$ there
exists $\tu y$ with $g(x,\tu y)=0$,  function $f = \SMIN(g)$ defined
by (safe) minimization
from $g$ is the (total) function $\N^{p} \to \N$ satisfying $
\SMIN(g) : \tu y \mapsto \min\{x; g(x,\tu y)=0\}$. 
\end{definition}

\begin{theorem}[Normal form for computable functions
  \cite{Kal43,Rose}] \label{normalform}
  Any total recursive function $f$ can be written as $f = g (
    \SMIN(h))$ for some elementary functions $g$ and $h$.
\end{theorem}

Consider the family of functions $E_n$  defined by induction
as follows. When $f$ is a function, $f^{[d]}$ denotes its $d$-th
iterate: $f^{[0]}(\vx)=x$, $f^{[d+1]}(\vx)=f(f^{[d]}(\vx))$:
\begin{eqnarray*}
\gE_0(x) &=& s(x) =  x+1, \\
\gE_1(x,y) &=& x+y, \\
\gE_2(x, y) &=& (x+1) \cdot (y+1),\\
\gE_3(x) &=& 2^x,\\
\gE_{n+1}(x)&=& \gE_n^{[x]}(1) \mbox{ for $n \geq 3$.}
\end{eqnarray*}


\begin{definition}[Bounded recursion ] \label{def:limclassique}
Given functions $g(\tu y): \N^p \to \N$ and $h(f,x,\tu
y): \N^{p+2} \to \N$ and $i(x,\tu y): \N^{p+1} \to \N$, the function
$f=\limREC(g,h)$ defined by bounded recursion from $g$ and $h$ is defined as the function $\N^{p+1} \to \N$
verifying
\begin{eqnarray*}
f(0,\tu y) &=&  g(\tu y) \\
f(x+1,\tu y) &=&h(f(x,\tu y),x,\tu y) \\
\mbox{
under the condition that: } \\
f(x,\tu y) &\le& i(x, \tu y).
\end{eqnarray*}
\end{definition}

\begin{definition}[\Greg{} hierarchy (see \cite{Rose})]
  Class $\E^0$ denotes the class that contains the constant function
  $\zero$, the projection functions $\projection{i}{p}$, the successor
  function $\sucs$, and that is closed under composition and bounded
  recursion.

  Class $\E^{n}$ for $n \geq 1$ is defined similarly except that
  functions max and $\gE_n$ are added to the list of initial functions.
\end{definition}

\begin{\propositions}[\cite{Odi92, Cam01}] \label{rqBsumbprod} Let $n
  \geq 3 $. A function is in class $\Gregn$ iff it belongs to the
  smallest set of functions that contains constant function $\zero$,
  the projection functions $\projection{i}{p}$, the functions
  successor $\sucs$, addition $\plus$, subtraction $\myominus$, and
  the function $\gE_{n}$ and that is closed under composition, bounded
  sum and bounded product.
\end{\propositions}

The above proposition means that closure under bounded recursion is
equivalent to using both closure under bounded sum and closure under
bounded product. Indeed, as explained in chapter $1$ of \cite{Rose}
(see Theorem 3.1 for details), bounded recursion can be expressed as a
minimization of bounded sums and bounded products, itself being
expressed as a bounded sum of bounded products.


The following facts are known:

\begin{\propositions}[\cite{Rose,Odi92,clote2013boolean}]
\begin{eqnarray*}
\Elem_3 &=& \Elem \subsetneq \PR 
\\ 
\Elem_{n} &\subsetneq& \Elem_{n+1} \mbox{ for $n \geq 3$} \\
\PR &=& \bigcup_{i} \Elem_i 
\end{eqnarray*}
\end{\propositions}

%
%

\subsection{Complexity theory and bounded schemes}


We suppose the reader familiar with the well-known complexity classes $\Ptime$
(polynomial time), $\NPtime$ or (non-deterministic polynomial time) or
$\Pspace$ (polynomial space). We denote by $\FPtime$ (resp. $\FPspaceN$) the class of functions, $f:\N^k\rightarrow \N$ with $k\in \N$, computable in polynomial time (resp. polynomial space) on deterministic Turing machines. Note that if $\FPtime$ is closed by composition, it is not the case of $\FPspaceN$ since the size of the output can be exponentially larger than the size of the input.

It turns out that the main complexity classes have  also been characterized
algebraically, by restricted form of recursion scheme.
A foundational result in that spirit is due to Cobham, who gave in
\cite{cob65} a characterization of function computable in polynomial time. The idea is to consider schemes similar to primitive
recursion, but with restricting the number of
induction steps.

Let $\succzero{.}$ and $\succun{.}$ be the successor functions defined by
$\succzero{x}=2.x$
and $\succun{x}=2.x+1$.

\begin{definition}[Bounded recursion on notations]
  A function $f$ is defined by bounded recursion scheme on notations
  from $g, h_0, h_1, k$, denoted by $f=\BRN(g,h_0,h_1)$, if 
\begin{eqnarray*}
 f(0, \tu y) &=& g(\tu y)\\
 f(\succzero{x}, \tu y) &=& h_0(f(x, \tu y) , x, \tu y) \mbox{ for $x \neq 0$} \\
f(\succun{x}, \tu y) &=& h_1(f(x, \tu y), x, \tu y) \\
\mbox{
under the condition that: } \\
  f(x, \tu y) &\le & k(x, \tu y) 
\end{eqnarray*}
for all $x,\tu y$. 
\end{definition}

Based on this scheme, Cobham proposed the following class of functions:

\begin{definition}[$\mathcal{F}_p$]
  The class $\mathcal{F}_p$ is the smallest class of primitive
  recursive functions containing $\zero$, the projections $\projection{i}{p}$, the
  successor functions $\succzero{x}=2.x$ and $\succun{x}=2.x+1$, the function
  $\# $ defined by $x \# y = 2 ^{\length{x} \times \length{y}}$ and closed by
  composition and by bounded recursion scheme on notations.
\end{definition}

This class turns out to be a characterization of polynomial time:

\begin{theorem}[\cite{cob65}, see \cite{Clo95} for a proof]
  $\mathcal{F}_p = \FPtime$.
\end{theorem}

Cobham's result opened the way to various characterizations of
complexity classes, or various ways to 
control recursion schemes. This includes the famous
characterization of $\Ptime$ from Bellantoni and Cook
in~\cite{bs:impl} and by Leivant in~\cite{Lei-LCC94}. Refer to
\cite{Clo95,clote2013boolean} for monographies presenting a whole
serie of results in that spirit. 


The task to capture $\FPspaceN$ is less easy since the principle of such characterizations is to use classes of functions closed by composition. However, for function with a reasonable output size some characterizations have been obtained.
Let us denote by $\myFPspace$, the class of functions of polynomial growth i.e. of functions $f:\N^k \rightarrow \N$, such that, for all $\tu x \in \N^k$, $\length{f(\tu x)}=O(\max_{1\leq i \leq k} \length{x_i})$. The following then holds:

\begin{theorem}[{\cite{thompson1972subrecursiveness},\cite[Theorem 6.3.16]{clote2013boolean}}]
A function over the integers is in $\myFPspace$ if and only if it belongs
  to the smallest set of functions that contains the constant function
  $\zero$, the projection functions $\projection{i}{p}$, the functions
  successor $\sucs$, function $\#$ 
  that is closed under composition and bounded recursion. 
\end{theorem}




\section{Computability and Discrete ODEs}\label{sec:Computability and Discrete ODEs}

Before coming back to efficient algorithms and complexity theory,  we consider functions defined by ODE under the prism of
computability.  This part is clearly inspired by ideas from \cite{Cam01,
  cam:moo:fgc:00}, but adapted here for our framework of discrete ODEs
that we believe to provide simpler explanations of statements of these
papers. 
Our settings in
particular avoid discussions related to how to deal with noise in
computations, as we are living in a world where computations are
exact. Furthermore, we
believe it
clearly helps the intuition of many of the constructions done in all
these references.

\subsection{About positive and negative integers and encodings}
As classical
computability is mainly dealing with functions over the natural
integers, i.e. over $\N$, while schemes with discrete ODEs naturally deals with
functions over the integers, i.e. over $\Z$, we need to fix some
conventions to be able to compare classes over the integers. Notice
that this is
very natural in our framework to consider functions that may
take negative values. 


\begin{definition}[Representation of integers]\label{def:representation of integers}
The set $\Z$  of integers can be encoded by the set $\{0,1\} \times \N$: couple $(s,n)$
with $s \in \{0,1\}$, $n \in \N$
encodes $(-1)^s n$.  Notice that $0$ corresponds both to $(0,0)$ and
$(1,0)$.
To avoid confusion, we will denote by $\tZ$ the set $\{0,1\} \times \N$.
\end{definition}


We will only deal with classes $\mathcal{C}$ of functions over
either the natural integer $\N$ or integers $\Z$.
We basically use the same convention for functions : let $f:\N^k\times
\Z^h \rightarrow \N^r \times \Z^s$, we denote by $\tilde{f}:\N^k\times
\tilde{\Z}^h \rightarrow \N^r \times \tilde{\Z}^s$ the function
equivalent to $f$ with above representation. 

Note that, if $x,y$ and $z$ are such that $x=y+z$ then $\tilde{x}=\tilde{y} \tilde{+} \tilde{z}$ and $\tilde{+}$ is primitive recursive. The same holds for multiplication and subtraction.

\subsection{Recursive and subrecursive classes of functions}


%
%
%

A first key remark is that at a computability level, many schemes can
actually be seen as particular natural types of ODEs.


\subsection{Subrecursive functions and discrete ODEs}

First, the purpose of this subsection is to observe that primitive
recursion is basically a discrete ODE schemata:

\begin{definition}[(Scalar) Discrete ODE schemata]  \label{def:dode2}
Given $g: \N^p \to \N$ and
  $h: \Z\times \N^{p+1} \to \Z$, we say that $f$ is
is defined by discrete ODE solving from $g$ and $h$, denoted by
$f=\ODE(g,h)$,  if $f: \N^{p+1} \to \Z$
corresponds to the (necessarily unique) solution of  Initial Value
Problem
\begin{oureqnarray} 
\label{eq:cauchy2}
\dderiv{f(x,\tu y)}{x} &=& h(f(x,\tu y),x,\tu  y)\\
f(0,\tu y)&=& g(\tu y ).
\end{oureqnarray}
\end{definition}

\begin{remark}
To be more general, we could take $g:\N^p \to \Z$. However, this would be of no use in the context of this paper.
\end{remark}

\begin{lemma}[Primitive recursion vs Discrete ODEs] \label{thm: primrecVsDODE}
  \begin{enumerate}
  \item Consider $g$ and $h$ as in Definition \ref{def:dode2} and $f=\ODE(g,h)$. Then
    $\tilde{f}$ is primitive recursive when $g$ and $\tilde{h}$ are. 
 When $f:\N^{p+1}\rightarrow \N$, then  $f$ is primitive recursive under the same conditions
  \item Consider $g$ and $h$ as in Definition
    \ref{def:classique}. Then $f=\REC(g,h)$ corresponds also to
    $f=\ODE(g,\overline{h})$ where $\overline{h}: \N^{p+2} \to \Z$ is
    defined by
    $$\overline{h}  (f(x,\tu y),x,\tu y)= h(f(x,\tu y),x,\tu y)-f(x,\tu y).$$
  \end{enumerate}
\end{lemma}

%
%
%

\begin{proof}
  For statement 1., applying the ODE schemata on primitive recursive
  functions $\tilde{h}$ and $g$, it  holds that: $f(0,\tu y) = g(\tu y)$ and
  $\tilde{f}(x+1,\tu y)=\tilde{h}(\tilde{f}(x,\tu y),x,\tu y)+
  \tilde{f}(x,\tu y)$, where addition is redefined to apply to elements of $\tilde{\Z}$. This is easily seen to be primitive recursive. When $f:\N^{p+1}\rightarrow \N$, one can extract $f(x,\tu y)$ from $\tilde{f}(x,\tu y)$   by a primitive recursive function.

  For statement 2., remark that $h(f(x,\tu y),x,\tu y)-f(x,\tu y) = f(x+1,\tu
  y) - f(x, \tu y) = \dderiv{f(x,\tu y)}{x}$. 
\end{proof}







Lemma~\ref{thm: primrecVsDODE} combined with Definition
\ref{def:classique} provides the following important characterization
of primitive recursive functions in terms of discrete ODEs. 

\begin{theorem}[A discrete ODE characterization of primitive recursive
  functions]
 The set of primitive recursive functions $\PR$ is the intersection
 with $\N^\N$ of the smallest set of functions that contains the zero
 functions $\zero$, the projection functions $\projection{i}{p}$, the addition and subtraction functions $\plus$ and $\minus$, and that is closed under composition and discrete $\ODE$
 schemata.
\end{theorem}

\subsection{Elementary functions, \Greg{} hierarchy and linear
  discrete ODEs}
\label{sec:greg}

Actually, this is even possible to be more precise, and provide a
characterization of the various subrecursive classes introduced up to
now. This part is clearly inspired from ideas from \cite{Cam01,
  cam:moo:fgc:00}, adapted here for discrete ODEs. 



This is very natural to 
restrict to linear
ODEs. This provides natural ways to talk about elementary functions and levels of the \Greg{}
hierarchy.


\begin{definition}[(Scalar) Linear ODE schemata]  \label{def:lode2}
Given $g: \N^p \to \N$, $a: \N^{p+1} \to \Z$ and $b: \N^{p+1} \to \Z$, we say that
$f$ is obtained by linear ODE solving from $g,a$ and $b$, denoted by $f=\LI(g,a,b)$,
if  $f: \N^{p+1} \to \Z$ corresponds to the (necessarily unique) solution of  Initial Value
Problem
\begin{oureqnarray}
\label{eq:lincauchy2}
\dderiv{f(x,\tu y)}{x} &=& a(x,\tu y) \cdot f(x,\tu y) + b(x,\tu y)\\
f(0,\tu y)&=& g(\tu y ).
\end{oureqnarray}
\end{definition}



First observe that bounded sums and products are of this specific
form: 

\begin{lemma}[Bounded sum]\label{bsum}
Let $k:\N^{p+1} \to \N$ be given.
Then
$f=\BSUMs(k)$ is the unique solution of initial value problem
\begin{eqnarray*}
\dderiv{f(x,\tu y)}{x} &=& k(x,\tu  y)\\
f(0,\tu y)&=& 0
\end{eqnarray*}
\end{lemma}

\begin{lemma}[Bounded product]\label{bprod}
Let $k:\N^{p+1} \to \N$ be given.
Then
$f=\BPRODs(k)$ is the unique solution of initial value problem
\begin{eqnarray*}
\dderiv{f(x,\tu y)}{x} &=& f(x,\tu y) \cdot (k(x,\tu  y)-1)\\
f(0,\tu y)&=& 1
\end{eqnarray*}
\end{lemma}

In the context of Ordinary Differential Equations, this is very natural
not to restrict to scalar functions, and the
following makes a clear natural sense.

\begin{definition}[Linear ODE schemata]  \label{def:lode3}
Given a vector $\tu G=(G_i)_{1 \le i \le k}$ 
matrix $\tu A=(A_{i,j})_{1 \le i,j \le k}$,
$\tu B= (B_i)_{1 \le i \le k}$ whose coefficients corresponds to functions
$g_i: \N^p \to \N^k$, and 
$a_{i,j}: \N^{p+1} \to \Z$ and $b_{i,j}: \N^{p+1} \to \Z$
respectively, we say that
$\tu f$ is obtained by linear ODE solving from $g,A$ and $B$, denoted
by $\tu f=\LI(\tu G,\tu A, \tu B)$,
if  $f: \N^{p+1} \to \Z^k$ corresponds to the (necessarily unique) solution of  Initial Value
Problem
\begin{oureqnarray}
\label{eq:lincauchy3}
\dderiv{\tu f(x,\tu y)}{x} &=& \tu A(x,\tu y) \cdot \tu f(x,\tu y) + \tu B(x,\tu y)\\
\tu f(0,\tu y)&=& \tu G(\tu y ).
\end{oureqnarray}
\end{definition}

One key observation behind the coming characterizations is the following:

\begin{lemma}[Elementary vs Linear ODEs]\label{elemVsLinODE}
  Consider $\tu G,\tu A$ and $\tu B$ as in Definition \ref{def:lode3}. Then
    $\tilde{\tu f}=\LI(\tu G,\tu A, \tu B)$ is elementary when $\tu G,
    \tilde{\tu A}$ and $\tilde{\tu B}$ are.
\end{lemma}

\begin{proof}
We do the proof in the scalar case, writing
$a,b,g$ for $\tu A,\tu B, \tu G$. The
general (vectorial) case follows from similar arguments. 
By Lemma~\ref{def:solutionexplicitedeuxvariables}, it follows that:

  \[f(x,\tu y) = \left( \prod_{t=0} ^{t=x-1} (1+a(t,\tu y)) \right) \cdot g(\tu y) +
b (x-1,\tu y) + 
  \sum_{u=0}^{x-2} \left( \prod_{t=u+1}^{x-1} (1+a(t,\tu y)) \right) \cdot b(u,\tu y).\]

  Clearly, $ \prod_{t=0} ^{t=x-1} (1+a(t,\tu y))=\BPRODs(1+a(t,\tu y))(x,\tu y)$. Similarly,

  \[
 p(u,x,\tu y) =^{def} \prod_{t=u+1}^{x-1} (1+a(t,\tu y))=\frac{\BPRODs(1+a(t,\tu y))(x,\tu y)}{\BPRODs(1+a(t,\tu y))(u+1,\tu y)}
  \]

  As the function $(x,y)\mapsto \lfloor x/y \rfloor$ is elementary from
  Lemma \ref{lem:derose}, we get that $p(x,\tu y)$ is elementary. 

  As multiplication is elementary, it follows that $$\sum_{u=0}^{x-2}
 p(x,\tu y)  b(u,\tu y)  = \BSUMs(p(u,x,\tu y) b(u,\tu y) )(x-2,\tu y)$$ is also
elementary , and  $\tilde{f}$ is elementary using closure by
composition and multiplication. 
\end{proof}

We get the following elegant characterization of the Elementary
functions in terms of Linear ODEs. 

\begin{theorem}[A discrete ODE characterization of elementary
  functions] \label{th:elem2}
  The set of elementary functions $\Elem$ is the intersection
  with $\N^\N$ of the smallest set of functions
  that contains the zero functions $\zero$, the projection functions $\projection{i}{p}$,
  the successor function $\sucs$, addition
  $\plus$, subtraction $\minus$, and that is closed under composition and
  discrete linear ODE schemata (respectively: scalar discrete linear
  ODE schemata) $\LI$.
\end{theorem}
	




Inspired by bounded recursion, this also makes sense to consider the
following (as expected, we write $\tu u \leq \tu v$ if it holds componentwise):

\begin{definition}[Bounded discrete ODE schemata]  \label{def:limdode}
Given $\tu g(\tu y): \N^p \to \N^k$ and
  $\tu h(\tu f,x,\tu y): \Z^k \times \N^{p+1} \to \Z^k$, and $\tu i(x,\tu y):\N^{p+1} \to \Z^k$,
 we say that $\tu f$ is
is defined by bounded discrete ODE solving from $\tu g$,$\tu h$ and $\tu i$, denoted by
$\tu f=\limODE(\tu g,\tu h,\tu i)$,  if $\tu f: \N^{p+1} \to \Z^k$
corresponds to the (necessarily unique) solution of  Initial Value
Problem
\begin{eqnarray*} \label{eq:cauchy3}
\dderiv{\tu f(x,\tu y)}{x} &=& \tu h(\tu f(x,\tu y),x,\tu  y)\\
\tu f(0,\tu y)&=& \tu g(\tu y ) \\
\mbox{ under the condition that: } \\
\tu f(x,\tu y ) &\le&  \tu i(x, \tu y)
\end{eqnarray*}
\end{definition}

\begin{lemma}[Primitive recursion vs Discrete ODEs]\label{primRecODEequ}
  \begin{enumerate}
  \item Consider $\tu g$, $\tu h$, $\tu i$ as in Definition \ref{def:limdode}. Then
    $\tilde{\tu f}=\limODE(\tu g,\tilde{\tu h},\tu i)$ is in $\E_n$ when $\tu g$ and $\tu h$
    and $\tu i$ are, and $n \ge 3$. 
  \item Consider  $\tu g$, $\tu h$, $\tu i$ as in Definition
    \ref{def:limclassique}. Then $\tu f=\limREC(\tu g,\tu h,\tu i)$ corresponds also to
    $\tu f=\limODE(\tu g,\overline{\tu h},\tu i)$ where $\overline{\tu h}: \N^{p+2} \to \Z^k$ is
    defined by
    $\overline{\tu h}  (\tu f,x,\tu y)= \tu h(\tu f,x,\tu y)-\tu f.$
  \end{enumerate}
\end{lemma}

\begin{proof}  For statement 1., this follows from exactly the same proof as
  for Lemma \ref{elemVsLinODE}. 

  Second item can be proved by observing that $\dderiv{\tu f(x,\tu y)}{x}
  = \overline{\tu h}(\tu f(x,\tu y),x,\tu y) = \tu f(x+1, \tu y) - \tu f(x, \tu y)$
  which is equal to $\tu h(\tu f(x, \tu y), x, \tu y) - \tu f(x, \tu y)$ by
  Definition~\ref{def:limdode}.
\end{proof}


This provide the following elegant characterizations of the levels of
the \Greg{} hierarchy in terms of bounded linear ODEs. 

\begin{theorem}[A discrete ODE characterization of  $\Gregn$ for $n \ge
  3$]
  For all $n \ge 3$, the set of functions in $\E_{n}$ is the smallest set of functions
  that contains $\gE_n$, constant function
  $\zero$, the projection functions $\projection{i}{p}$, the functions successor $\sucs$,
 and that is closed under composition and $\limODE$.
\end{theorem}




\begin{proof}
  Using Theorem~\ref{rqBsumbprod}, this follows from Lemmas
  \ref{bsum}, \ref{bprod} and \ref{primRecODEequ}.
\end{proof}



\subsection{Computability and discrete ODEs}

If we want to talk about computable functions, and not only about
subrecursive functions, a first method is to add directly minimization
to considered operators.


\subsubsection{By adding a minimization operator}

\begin{theorem}[A discrete ODE characterization of total recursive
  functions]
  The set of total recursive functions $\Elem$ is the intersection
  with $\N^\N$ of the smallest set of functions
  that contains the zero functions $\zero$, the projection functions $\projection{i}{p}$,
  the successor function $\sucs$, addition
  $\plus$, subtraction $\minus$, and that is closed under composition and
  discrete (even linear) ODE schemata $\LI$, and safe minimization. 
\end{theorem}

\begin{proof}
One direction follows from Theorem \ref{th:elem2} (characterization of
elementary functions) and Theorem \ref{normalform} (normal form
theorem) in one
direction. And from a clear generalization of previous arguments in
the other direction.
\end{proof}


\subsubsection{By programming minimization}


But actually, minimization can be programmed using discrete ODEs \emph{in some
	sense}. 
Indeed, minimization can be programmed in the following sense.

\begin{theorem}[Programming Minimization]
	Consider a function $g: \N^{p+1} \to \N$. 
	Then the solution of initial value problem 
	\begin{eqnarray*}
		f(0,\tu y) &=& 0 \\
		\dderiv{f(x,\tu y)}{x}  &=& \cond{g(f(x,\tu y))}{1}{0} \\
	\end{eqnarray*}
	\noindent is such that for all $\tu y$, $f(x,\tu y)$ is eventually a 
	constant $k=k(\tu y)$ when $x$ increases if and only if there is some
	$x$ with $g(x,\tu y)=0$. This constant $k(\tu y)$ corresponds to $\SMIN(g)(\tu y)$ for all $\tu
	y$. 
\end{theorem}

This leads to the following natural concept: The idea is that $\SMIN(g)$ is computable in
the following  sense considering $h_1(x,\tu y) = f(x,\tu y)$ and
$h_2(x,\tu y)=\signcompn{g(f(x,\tu y))}$.

\begin{definition}[Discrete ODEs as a computational model]
	We say that a total function $f: \N^p \to \N$ is ODE computable if there
	exist some function $h_1,h_2: \N^{p+1} \to \N^2$ in the smallest set of functions
	that contains the zero functions $\zero$, the projection functions $\projection{i}{p}$,
	the successor function $\sucs$, and that is closed under composition and
	discrete  $\ODE$ schemata  such that:
	for all $\tu y$, 
	\begin{itemize}
		\item there exists some $T=T(\tu y)$ with
		$h_2(T,\tu y) \neq 0$;
		\item $f(\tu y) = h_1 (T,\tu y)$ where $T$ is the smallest such $T$. 
	\end{itemize}
\end{definition}

The following is then easy to establish: 

\begin{theorem}[Discrete ODE computability = classical computability] 
	A total function $f$ is ODE computable if and only if it is total
	recursive.
\end{theorem}

\section{Restricted recursion and integration schemes}
\label{sec:restrict}

In order to talk about complexity instead of computability, we need to
put some restrictions on integrations schemes.

%
%
{
\begin{remark} Observe that this is necessary. Indeed, the solution of a polynomial ordinary differential
	equation (ODE) can grow very very fast.

	Indeed:
	\begin{eqnarray*}
		\left(\fallingexp{x}\right)^\prime &=& \fallingexp{x} \\
		\left(\fallingexp{\fallingexp{x}}\right)^\prime &=& \fallingexp{x} 
		\cdot \fallingexp{\fallingexp{x}} \\
		\left(\fallingexp{\fallingexp{\fallingexp{x}}}\right)^\prime &=&
		\fallingexp{x} \cdot
		\fallingexp{\fallingexp{x}} \cdot
		\fallingexp{\fallingexp{\fallingexp{x}}}
		\\
		&\vdots& \\
	\end{eqnarray*}
	and so on,
	is solution of degree 2 polynomial ODE:
	
	\begin{eqnarray*}
		y^\prime_1&=& y_1 \\
		y^\prime_2&=& y_1 \cdot y_2 \\
		y^\prime_3&=& y_2 \cdot y_3 \\
		&\vdots&
	\end{eqnarray*}
	with initial condition $y_1(0)=y_2(0)=y_3(0)=\dots=1$.
That means that if we consider a two general integration
        scheme, then we get such towers of exponentials. Clearly, such
        a function is not polynomial time computable, as only writing
        its value in binary cannot be done in polynomial time. 
      \end{remark}
    }
%
%
We propose to introduce the following variation on the notion of
derivation: derivation along some function $\lengt(x, \tu y)$. 



\begin{definition}[$\lengt$-ODE] Let $\lengt:\N^{p+1} \rightarrow \Z$. We  write
	\begin{equation}\label{lode}
	\dderivL{\tu f(x,\tu y)}= \dderiv{\tu f(x,\tu y)}{\lengt(x,\tu
          y)} = \tu h(\tu f(x,\tu y),x,\tu y),
	\end{equation}

as a formal synonym for
$$ \tu f(x+1,\tu y)= \tu f(x,\tu y) + (\lengt(x+1,\tu y)-\lengt(x,\tu y)) \cdot
\tu h(\tu f(x,\tu y),x,\tu y).$$
\end{definition}
%
\begin{remark}
This is motivated by the fact that the latter expression is similar to
classical formula for classical continuous ODEs:
$$\frac{\delta f(x,\tu y )}{\delta x} = \frac{\delta
  \lengt (x,\tu y) }{\delta x} \cdot \frac{\delta f(x,\tu
  y)}{\delta \lengt(x, \tu y)}.$$ 
\end{remark}
This will allow us to simulate suitable change
of variables using this analogy.
We
will talk about $\lengt$-IVP when some initial condition is added. An important special case is when  $\lengt(x, \tu y)$ corresponds to the
length $\lengt(x,\tu y)=\length{x}$
function: we will call this special case length-ODEs.

\begin{example}[Example \ref{ex:some} continued]
  The trick used in Example \ref{ex:some} can be read as using a new parameter
  $t=\length{x}$ logarithmic in $x$, using relation
  	$$ \dderiv{\fonction{some}_h(x)}{\length{x}} = E(\fonction{some}_h(x),\length{x},x).$$
      \end{example}

\begin{example}[Function $2^{\length{x}}$ and  $2^{\length{x}^2}$]
%
To compute function $f: x \mapsto
2^{\length{x}}$, 
a method would consist in computing $f(x)$ using the fact  $f(x) =
\dint{0}{x}{f'(t)}{t}$ which would \textit{a priori}  
requires time 
  $x$;
 %
But a  more efficient method consists in stating that
$f(x)=F(\length{x})$ where $F(t)=2^t$ is solution of IVP 
%
$F'(t)= F(t)$, $F(0)=1$. 
This is a fast (polynomial) algorithm to solve our problem. Once
again, 
we have used a change of variable in order to
compute faster.
Thinking about what we have just done, we have basically observed the
fact that
\[
    (2^{\length{x+1}} ) '= 
 \length{x}'
 \cdot 2^{\length{x}}
 \mbox{ that is to say   } 
 \dderiv{2^{\length{x}}}{\length{x}} = \length{x}' \cdot 2^{\length{x}}
\]
%

This is what leaded us to consider change of variable $t=
 \length{x}$ and what leaded to above more efficient algorithm,
 considering $F(t)$ instead of $f(x)$, with similarities with the
 relation for continuous derivative 
 $\frac{\delta f(x )}{\delta t} = \frac{\delta
   t}{\delta x} \cdot \frac{\delta f(x)}{\delta t}$.


Suppose now that we want to compute function $f: x \mapsto
2^{\length{x}^2}$.
We can use the same principle, observing that 
$$
  (2^{\length{x+1}^2 } ) '= 
 (\length{x+1}^2-\length{x}^2) 
 \cdot 2^{\length{x}^2}
 $$
 \mbox{ that is to say   }
 $$
 \dderivL{ 2^{\length{x}^2}}= \left(\length{x}^2\right)' \cdot
2^{\length{x}^2}
\mbox{ considering } \lengt(x) = \length{x}^2
$$
and then noticing that $f$ is consequently computed fast (in polynomial time) as
$F(\length{x}^2)$.

\end{example}

%



{
\begin{example} 
$f(x,y)=2^{\length{x}\cdot \length{y}}$ is the solution of the
following length-IVP: 
\begin{eqnarray*}
f(0, y)&=&2^{|y|}  \\
\dderivl{f(x, y)} &=& f(x,y)\cdot (2^{\length{y}}-1), 
\end{eqnarray*}


\noindent since $2^{\length{x+1}\cdot \length{y}}=2^{\length{x}\cdot
  \length{y}} +  \length{x}' \cdot
2^{\length{x}\cdot \length{y}}\cdot (2^{\length{y}}-1).$

\end{example} 
}

%

\subsection{General theory}

The following result 
though simple, illustrate one key property of the $\lengt$-ODE scheme under a computational point of view: it's dependence on the number of distinct values of function $\lengt$.


\begin{definition}[$Jump_\lengt$]
Let $\lengt(x,\tu y)$ be some function.
Fixing $\tu y$, we write $Jump_\lengt(x,\tu y)=\{0 \le i \le x-1 | \lengt(i+1,\tu y) \neq \lengt(i,\tu
y)\}$
(that is to say the set of points where $\lengt$ has a value that changes)
and 	$\alpha:[0..|Jump_\lengt(x,\tu y)|-1]\rightarrow Jump_\lengt(x,\tu y)$ for an
increasing function enumerating these points:  If  $i_0 < i_1 < i_2 <
\dots < i_{card(Jump_\lengt(x,\tu y))-1}$ denote
        all elements of $Jump_\lengt(x,\tu y)$, then
$\alpha(j)=i_j\in Jump_\lengt(x,\tu y)$.
\end{definition}


\begin{lemma}[Fundamental Observation]\label{lem:fundamental observation} Let $k\in \N$, $f: \N^{p+1}\rightarrow \Z^d$ and 
$\lengt:\N^{p+1}\rightarrow \Z$  be some functions.
	Assume that \eqref{lode} holds.
Then:
	$$\tu f(x,\tu y) =
	\tu f(0,\tu y) + \dint{0}{card(Jump_\lengt(x,\tu y))}{ 
          {\Delta 
          \lengt(\alpha(u),\tu y)}
        \cdot \tu h(\tu f(\alpha(u),\tu y),\alpha(u),\tu y)}{u}$$

    \end{lemma}

\begin{proof} 
	By definition, we have

	$$\tu f(x+1,\tu y) = \tu f(x,\tu y) + (\lengt(x+1,\tu y)-\lengt(x,\tu
        y)) \cdot \tu h(\tu f(x,\tu y),x,\tu y).$$

	Hence,
	\begin{itemize}
		\item as soon as $i \not\in Jump_\lengt(x,\tu y)$, then $\tu
                  f(i+1,\tu y)=\tu f(i,\tu y)$, since
		$\lengt(i+1,\tu y)-\lengt(i,\tu y)=0$. In other words, $$\Delta
                \tu f(i,\tu y)=0.$$
		\item as soon as $i \in Jump_\lengt(x,\tu y)$, say $i=i_j$,
                  then $$\Delta \tu f(i_j,\tu y) = (\lengt(i_j+1
		,\tu y)-\lengt(i_j,\tu y)) \cdot \tu h(\tu f(i_j,\tu y),i_j,\tu y)$$
		I.e.

		$$
		\Delta \tu f(i_j,\tu y) =
		\Delta \lengt(i_j,\tu y) \cdot \tu h(\tu f(i_j,\tu y),i_j,\tu y) $$
	\end{itemize}

	Now  

\begin{eqnarray*}
\tu f(x,\tu y) &=& \tu f(0,\tu y) + \dint{0}{x}{\Delta
          {\tu f(t,\tu y)}}{t}  \\
&=& \tu f(0,\tu y) +
	\sum_{t=0}^{x-1} \Delta {\tu f(t,\tu y)} \\
&=& \tu f(0,\tu y) + \sum_{i_j
          \in Jump_\lengt(x,\tu y)} \Delta
	{\tu f(i_j,\tu y)} \\
&=& \tu f(0,\tu y) +  \sum_{i_j \in Jump_\lengt(x,\tu y)} \Delta \lengt(i_j,\tu y) \cdot
	\tu h(\tu f(i_j,\tu y),i_j,\tu y)  \\
&=& \tu f(0,\tu y) +  \sum_{j=0}^{card(Jump_\lengt(x,\tu y))-1} \Delta \lengt(\alpha(j),\tu y) \cdot
	\tu h(\tu f(\alpha(j),\tu y),\alpha(j),\tu y) \\
&=&  \tu f(0,\tu y) + \dint{0}{card(Jump_\lengt(x,\tu y))}{\Delta
          \lengt(\alpha(u),\tu y) \cdot \tu h (\tu f(\alpha(u),\tu
          y),\alpha(u),\tu y)}{u}
\end{eqnarray*}
which corresponds to the expression. 
\end{proof}


%
%
%
%

The proof of the Lemma is based on (and illustrates) some fundamental
aspect of $\lengt$-ODE from their definition: for fixed $\tu y$, the value of $\tu f(x,\tu y)$ only
changes when the value of $\lengt(x,\tu y)$ changes. This implies that
the value of $\tu f(x,\tu y)$ must then depend  on $\tu y$ and
$\lengt (x,\tu
y)$.
{
We formalize this in the following definition.

\begin{definition}[$\lengt$-expressiveness] 	
Let $k\in \N$, $f: \N^{p+1}\rightarrow \Z^d$ and 
$\lengt:\N^{p+1}\rightarrow \Z$  be some functions. 
	We say that $\tu f(x,\tu y)$ is $\lengt$-expressible if
	there exists some function $\tu g:\N^{p+1}\rightarrow \Z$ such
        that $\tu f(x,\tu y)=\tu g(\lengt(x,\tu y),\tu y)$.
\end{definition}


\begin{corollary}
Let $k\in \N$, $f: \N^{p+1}\rightarrow \Z^d$ and 
$\lengt:\N^{p+1}\rightarrow \Z$  be some functions as above. 
Then $\tu f(x,\tu y)$ is $\lengt$-expressible.
\end{corollary}

} Let's make a pause to ponder. From the above results, if $\lengt$ is chosen such that
 $|Jump_\lengt(x,\tu y)|=|\{0 \le
 i \le
 x - 1| \lengt(i+1,\tu y) = \lengt(i,\tu y)\}| \leq P(\length{x}, \length{\tu y})$
  for some polynomial $P$ then, the number of distinct values
 of $\tu f(x',\tu y)$ with $x'\leq x$ that are necessary to compute
 $\tu f(x,\tu y)$ is polynomial in $\length{x}$ and $\length{\tu y}$. Hence, at least in
 terms of the number of steps (not necessarily in terms of the size of
 the intermediate objects), $\tu f(x,\tu y)$ can be computed fast.


\subsection{Fundamental alternative view}
%
%
If previous hypotheses hold, there is then an alternative view to
understand the integral, by using a change of variable, and by
building a discrete ODE that mimics the computation of the
integral.
Basically, we are using the fact that we can consider some
parameter $t$ corresponding to $\lengt(x,\tu y)$.
Indeed:


\begin{lemma}[Fundamental alternative view] \label{fundob}
Let $k\in \N$, $f: \N^{p+1}\rightarrow \Z^d$,
$\lengt:\N^{p+1}\rightarrow \Z$  be some functions and assume that \eqref{lode} holds.
Then $\tu f(x,\tu y)$ is given by 
$\tu f(x,\tu y)= \tu F(\lengt(x,\tu y),\tu y)$
%
where $\tu F$ is the solution of initial value problem
\begin{eqnarray*}
\dderiv{\tu F(t,\tu y)}{t} &=& 
                               {\Delta\lengt(t,\tu y)}
                               \cdot
\tu h(\tu F(t, \tu y),t,\tu y) \\ 
\tu F(0,\tu y)&=&\tu f(\lengt(0,\tu x),\tu y). 
\end{eqnarray*} 

We will say in that case the IVP is converging ``in time $\lengt(x,\tu y)$''
to $\tu f(x,\tu y)$. 
Conversely, if there is such a function $\tu F$, then a discrete ODE
of the type of \eqref{lode} can easily be derived.

\end{lemma}





{
\begin{example}
  The previous discussion about the complexity of computing  $x \mapsto
  2^{\length{x}}$ and $x \mapsto 2^{\length{x}^2}$ is a 
  concrete applications of all these remarks.
\end{example}
}

  \begin{example} 
  Let us consider an example, where $\lengt(x)$ is not
  $\length{x}$  (or a power of it):  Suppose we want to compute
  $f: x \mapsto 2^{ \lfloor \sqrt{x} \rfloor }$: 
%
%
%
%
  Consider 
  $\lengt(x) = \lfloor \sqrt{x} \rfloor$. 
We have $$ \dderiv{f(x)}{\lengt(x)} = \lengt'(x) \cdot f(x) = \left(
    \lfloor \sqrt{x+1} \rfloor -
\lfloor \sqrt{x} \rfloor \right) \cdot f(x). 
$$








One may think that the number
$|Jump_\lengt(x)|$  of $\lengt$, i.e. the number of  jumps 
  of factor $(\lfloor \sqrt{x+1} \rfloor -
  \lfloor \sqrt{x} \rfloor)$  is hard to predict, but the point is to
  look at the method we devised to compute $ \lfloor \sqrt{x} \rfloor$
  in Example \ref{sec:racine}: It is basically expressing $\lfloor \sqrt{x}
  \rfloor$ as some function $G$ of $\fonction{some}_h(x)$: We wrote $\lfloor \sqrt{x}
  \rfloor = G( \fonction{some}_h(x))$ for some function $G$.
 Consequently, we could
also consider variable $\lengt_2(x)=\fonction{some}_h(x)$, and  see
from expressions that the
number of jumps $|Jump_\lengt(x)|$   of previous $\lengt$ is actually
related to the $|Jump_{\lengt_2}(x)|$  of
this new $\lengt_2(x)$.
We also have
$$ \dderiv{f(x)}{\lengt_2(x)} = \lengt_2'(x) \cdot f(x) = \left(
    G( \fonction{some}_h(x+1))- G( \fonction{some}_h(x) \right) \cdot f(x). 
$$

Observing
that $\fonction{some}_h(x)$ is  in turn
computed in ``time'' $\length{x}$ using the method of Example
\ref{sec:racine},  the number of jumps
for all these $\lengt(x)$ are always polynomials, and we are guarantee
that all these
expressions lead to fast (polynomial) algorithms.  
\end{example}




\begin{remark}
  This method clearly extends to more general functions: Generalizing
  the above reasoning, we can compute fast functions of type $x \mapsto g(
  \lfloor \sqrt{x} \rfloor)$ as soon as we have a fast ODE computing
  $g$. Similarly, $\lfloor \sqrt{x} \rfloor$ can be replaced by
  anything that can be computed fast basically using similar
  techniques.
  \end{remark}


{
\subsection{Length-ODEs}

An important and natural case is the special case where $\lengt(x,\tu y)$ is the
usual one variable length function $\lengt(x,\tu y)=\length{x}$.
We will of course write $\dderivl{\tu f(x,\tu y)}
$ in that
case for $\dderivL{\tu f(x,\tu y)}$. 

We can adapt the Lemma above to
this special case of a what we will call length-ODE. Namely:

\begin{corollary}[First view]~\label{corollary:fundamental observation}
	Let $\lengt:\N\rightarrow \N$ be defined by $\lengt(x)=\length{x}$ for all integer $x$ and $f$ satisfies the hypothesis of Lemma~\ref{lem:fundamental observation}. Then,

	$$\tu f(x,\tu y) =
\tu	f(0,\tu y) + \dint{0}{\length{x}}{\tu h(\tu f(2^u-1,\tu y),2^u-1,\tu y)}{u}$$

	\noindent Or, equivalently:

	$$\tu f(x,\tu y) = \tu f(0,\tu y) + \sum_{i=0}^{\length{x}-1}
        \tu h(\tu f(2^i-1,\tu y),2^i-1,\tu y)$$
\end{corollary}

\begin{proof} Immediate consequence of Lemma~\ref{lem:fundamental observation}. Function $\alpha$ is such that $\alpha(i)=2^i-1$.
\end{proof}


\begin{corollary}[Alternative view]

	Let $\lengt:\N\rightarrow \N$ be defined by $\lengt(x)=\length{x}$ for all
        integer $x$ and $f$ satisfies the hypothesis of
        Lemma~\ref{lem:fundamental observation}. 
Then    $\tu f(x,\tu y)$ is given by $\tu f(x, \tu y) =
F(\length{x},\tu y)$
\noindent where $\tu F$ is the solution of initial value problem


\[	\dderiv{\tu F(t,\tu y)}{t} = \tu h ( \tu F(t,\tu y ),t,\tu y) 
\mbox{ with }	\tu F(0,\tu y)=\tu f(0,\tu y)
\]

%
\end{corollary}

In other words, for $\lengt(x)=\length{x}$, this offers us also two ways to present a length-ODE for
a function $f(x,\tu y)$: either by considering equation of the type
of~\eqref{lode} or by considering $\tu f(x,\tu y) = \tu
F(\length{x},\tu y)$ where $\tu F$ given by an equation of the form:

\begin{equation}\label{lode2}
\dderiv{\tu F(t,\tu y)}{t} = \tu h(\tu F(t,\tu y),t,\tu y)
\end{equation}
with $\tu F(0,\tu y)=\tu f(0,\tu y)$. As before, the idea is that  $t$ is a parameter logarithmic in $x$, namely
$t=\length{x}$.

}

\newcommand\vectorp[2]{\left(\begin{array}{l}
#1 \\ #2 \\
\end{array} \right)}

Our purpose now is to discuss which kind of problems can be solved
efficiently using similar techniques: it  turns out to be exactly all of $\FPtime$
It will be made clear from the incoming discussion and results. 

\subsection{Linear length-ODEs}


\begin{remark}
In all previous reasoning, we considered that a function over the integers is polynomial time
computable if it is in the length of all its arguments, as this is the
usual convention. When not
explicitely stated, this is our convention.  
As usual, we also say that some vectorial function (respectively:
matrix) is polynomial time computable if all its components are.
We will need sometimes to consider also polynomial dependency directly
in some of the variables and not on their length: This happens in the
next fundamental lemma.
\end{remark}
\newcommand\norm[1]{\| #1 \|}
We write $\norm{\cdots}$ for the sup norm: given some matrix $\tu
A=(A_{i,j})_{1 \le i \le n, 1 \le j \le m}$, 
$\norm{A}=\max_{i,j}
A_{i,j}$.



\begin{lemma}[Fundamental observation] \label{fundamencore}
Consider ODE 
\begin{equation} \label{eq:bc}
\tu f’(x,\tu y)=  {\tu A} ( \tu f(x,\tu y),
x,\tu y) \cdot
  \tu f(x,\tu y)
  +   {\tu B} ( \tu f(x,\tu y),
  x,\tu y).
\end{equation}
Assume:
\begin{enumerate}
\item Initial condition $\tu G(\tu y) = ^{def}
  \tu f(0, \tu y)$, as well as Matrix $\tu A$
  and vector $\tu B$ are polynomial time computable.
\item $\length{ \norm{{\tu A} ( f, x,\tu y)} } \le 
  \length{\norm{\tu f}} + p_{\tu A}(x,\length{\tu y})$ for some polynomial $p_A$
\item $\length{ \norm{{\tu B} ( f, x,\tu y)} } \le 
  \length{\norm{ \tu f}} + p_{\tu B}(x,\length{\tu y})$ for some polynomial $p_B$


\end{enumerate}

Then its solution $\tu f(x, \tu y)$ is polynomial time computable in $x$ and the length
of $\tu y$. 
\end{lemma}

\begin{proof}
We know by Lemma \ref{def:solutionexplicitedeuxvariables} that we
must have:

\begin{equation} \label{autreec}
\tu f(x,\tu y) =  \left( \fallingexp{\dint{0}{x}{\tu A( \tu f(t,\tu y) ,
    t,\tu y)}{t}} \right) \cdot \tu G (\tu y)+
\dint{0}{x}{ \left(
  \fallingexp{\dint{u+1}{x}{\tu A(\tu f(t,\tu y) ,
      t,\tu
      y)}{t}} \right) \cdot \tu B( \tu f(u,\tu y) ,
  u,\tu y)}{u}.
\end{equation}


The key point is that Equation \eqref{autreec} provides a (recursive)
algorithm to compute $\tu f(x, \tu y)$ for all $x$.
To see it, it may help to see that this can also be expressed as

\begin{equation} \label{autreecc}
\tu f(x, \tu y) =
	\sum_{u=-1}^{x-1 } \left(
	\prod_{t=u+1}^{x-1} (1+\tu A( \tu f(t,\tu y),
        t,\tu y)) \right) \cdot \tu B(\tu  f(u,\tu y) ,
        u,\tu y).
\end{equation}

with the conventions that $\prod_{x}^{x-1} \tu \kappa(x) = 1$ and $\tu
B( \cdot ,
-1,\tu y)=\tu G(\tu y)$.

Clearly the number of arithmetic operations to evaluate $\tu f(x, \tu y)$
by this method is polynomial in $x$: basically we have to sum $x+1$
terms, each of them involving at most $x-1$ multiplications. This can
be done in the requested complexity if we are sure that the size of
the involved quantities remains polynomial in $x$ and the length of
$\tu y$.

Since the length of $\tu B(\tu f(u,\tu y) , u,\tu y)$ and of
$\tu A( \tu f(t,\tu y), t,\tu y)$ is at most polynomial in the length
of $\tu f(t,\tu y)$ we only need to be convinced that the size of
$\tu f(x, \tu y)$ remains polynomial. 
But it holds, as from \eqref{eq:bc} we get

$$ \tu f(x+1, \tu y) =  (1 + \tu A(\tu f(x, y), x , \tu y))  \cdot \tu
f(x,\tu y) 
+ \tu B(\tu f(x, \tu y), x, \tu y) $$

\noindent and hence

\begin{eqnarray*}
 \length{ \tu f(x+1, \tu y) } &\le& \max ( \length{(1 + \tu A(\tu f(x, y), x , \tu y))} +
\length{\tu f(x,\tu y)} ,  \length{B(\tu f(x, \tu y), x, \tu y)} ) \\
& \le &  \length{\tu f(x, \tu y)} + p_{\tu f}(x, \length{\tu y}) + 1
\end{eqnarray*}

\noindent for polynomial $p_f$, that we may assume
without loss of generality to be increasing in its first argument.
It follows from an easy induction that we must have

$$\length{\tu f(x, \tu y) } \le \length{G(\tu y)} + x \cdot p_{\tu f}(x, \length{ \tu y}).$$





\end{proof}


We now go to specific forms of linear ODEs. 








\newcommand\polynomial{ \fonction{sg}-polynomial}



\begin{definition}
A \polynomial{}  expression $P(x_1,...,x_h)$ is a expression built-on
$+,-,\times$ (often denoted $\cdot$) and $\sign{}$ functions over a set of variables $V=\{x_1,...,x_h\}$ and integer constants.
The degree $\deg(x,P)$ of a term $x\in V$ in $P$ is defined inductively as follows:
\begin{itemize}
	\item $\deg(x,x)=1$ and for  $x'\in X\cup \Z$ such that $x'\neq x$, $\deg(x,x')=0$
	\item $\deg(x,P+Q)=\max \{\deg(x,P),\deg(x,Q)\}$
\item $\deg(x,P\times Q)=\deg(x,P)+\deg(x,Q)$
\item $\deg(x,\sign{P})=0$
\end{itemize}
A \polynomial{}  expression $P$  is \textit{essentially constant} in
$x$ if $\degre{x,P}=0$. 



\end{definition}

%

Compared to the classical notion of degree in polynomial expression,
all subterms that are within the scope of a sign function contributes
for $0$ to the degree.
%
%
A vectorial function (resp. a matrix or a vector) is said to be a \polynomial{} expression if all
its coordinates (resp. coefficients) are. 
it is said to be
\textit{essentially constant} if all its coefficients are.

A (possibly vectorial) \polynomial{} expression $\tu g(\tu f(x, \tu y), x,
\tu y)$ is said to \textit{essentially linear} in $\tu f(x, \tu y)$ if
it is of the form $$
\tu g(\tu f(x, \tu y), x,
\tu y) =
\tu A [\tu f(x,\tu y), \tu h(x,\tu y),
x,\tu y]  \cdot \tu f(x,\tu y) + \tu B [\tu f(x,\tu y), \tu h(x,\tu y),
x,\tu y ] 
$$
where $\tu A$ and $\tu B$ are \polynomial{} expressions essentially
constant in $\tu f(x, \tu y)$.

%


\begin{example}
	The expression $P(x,y,z)=x\cdot \sign{(x^2-z)\cdot y} + y^3$ is linear in $x$, essentially constant in $z$ and not linear in $y$.
	The expression $P(x,2^{\length{y}},z)=\sign{x^2 - z}\cdot z^2 + 2^{\length{y}}$ is essentially constant in $x$, essentially linear in $2^{\length{y}}$ (but not essentially constant)  and not essentially linear in $z$. The expression:
%
	$
	\cond{x}{y}{z}=y +
	\signcomp{x}\cdot (z-y)=
	y + (1-\sign{x})\cdot (z-y)
	$
        is essentially constant in $x$ and linear in $y$ and $z$.
\end{example}

\begin{definition}
Function $\tu f$ is linear $\lengt$-ODE definable (from $\tu u$ and
$\tu g)$ if it corresponds to the
solution of $\lengt$-IVP
\begin{oureqnarray}\label{SPLode}
\dderivL{\tu f(x,\tu y)}&=&   \tu u(\tu f(x,\tu y), \tu h(x,\tu y),
x,\tu y) \\
f(0,\tu y) &=& \tu g(\tu y) 
\end{oureqnarray}
\noindent where $\tu u$ is \textit{essentially linear} in $\tu f(x, \tu y)$. When $\lengt(x,\tu y)=\length{x}$, such a system is called linear length-ODE.
\end{definition}

The previous statements lead to the following:

\begin{lemma}[Fundamental Observation for linear $\lengt$-ODE]~\label{lem:fundamental observation linear length ODE}
	Assume that $\tu f(x,\tu y)$ is solution of \eqref{SPLode}. 
        Then $\tu f(x, \tu y)$ 
	can be computed in polynomial time 
	under the following conditions:
	\begin{enumerate}
		\item $\tu f(0, \tu y)= \tu g(\tu
                  y)$ is computable in polynomial-time. 
		\item\label{fund obs cond 2}  function $\tu h$  is computable in polynomial
		time. 

		\item\label{fund obs cond 13} there exist $c\in \N$, such that, for each
                  $\tu y$, $|Jump_\lengt(\tu y)|\leq \length{x}^c$.
	\end{enumerate}

\end{lemma}

\begin{proof}
Thanks to condition~\label{fund obs cond 3} above, we can replace
parameter $x$ and derivation in $\lengt(x,\tu y)$ by a parameter $t\leq
\length{x}^c$ and derivation in $t$ by Lemma \ref{fundob}. 

This leads to an ODE of the form:
$$
\tu f’(x,\tu y)= \overline {\tu A} ( \tu f(x,\tu y),
x,\tu y) \cdot
  \tu f(x,\tu y)
  +   \overline {\tu B} ( \tu f(x,\tu y),
  x,\tu y).
$$
by setting 
\begin{eqnarray*}
\overline {\tu A} ( \tu f(x,\tu y),
x,\tu y) &=& {\tu A} ( \tu f(x,\tu y), h(x, \tu y),
x,\tu y) \\
\overline {\tu B} ( \tu f(x,\tu y),
  x,\tu y) &=&
{\tu B} ( \tu f(x,\tu y), h(x, \tu y),
x,\tu y) 
\end{eqnarray*}

But then Lemma \ref{fundamencore} applies, and we get precisely the
conclusion, observing that the fact that the
corresponding matrix $\overline {\tu A}$ and vector $\overline {\tu
  B}$ are essentially constant in $\tu f(x, \tu y)$ guarantees
hypotheses of Lemma \ref{fundamencore}. 
\end{proof}

 \section{A characterization of polynomial time}
 \label{sec:A characterization of polynomial time}

\subsection{Register machines}



A register machine program (a.k.a. \textsf{goto} program) is a finite sequence of ordered labeled instructions acting on a finite set of registers of one of the following type:

\begin{itemize}
	\item increment the $j$th register $R_j$ by the value of $k$th
          register $R_k$ and go the next instruction:
	 $$
	R_j:=R_j+R_k$$
	\item decrement the $j$th register $R_j$ by the value of $k$th
          register $R_k$ and go the next instruction: $$
	R_j:=R_j-R_k$$

	\item set  the $j$th register $R_j$ to integer $k$, for $\ell \in
          \{0,1\}$ and go the next instruction: $$
	R_j:=k$$

	\item if register $j$ is equal to $0$, go to instruction $p$ else go to next instruction.

	 $$
	\mathsf{if} R_j = 0 \ \textsf{goto} \ p$$

	\item halt the program:
	$
	\mathsf{halt}
	$
\end{itemize}

In the following, since coping with negative numbers on classical
models of computation can be done through simple encodings, we will
not restrict ourself to non-negative numbers. 

	\begin{definition} Let  $t:\N \rightarrow \N$.
	A function $f:\N^p\rightarrow \Z$ is computable in time  $t$
        by a register machine $M$ with $k$ registers if:
	\begin{itemize}
		\item when starting in initial configuration with registers $R_1,\dots, R_{\min(p,k)}$ set to $x_1,\dots,x_{\min(p,k)}$ and  all other registers to $0$ and
		\item starting on the first instruction (of label $0$),
	\end{itemize}

\noindent Machine $M$ ends its computation after at most
$t(\length{\tu x})$ instructions where $\length{\tu x}=\length{x_1}+\cdots+\length{x_p}$ and with register $R_0$ containing $f(x_1,\dots,x_p)$.

A function is computable in polynomial time by $M$ if there exists
$c\in \N$ such that $t(\length{\tu x})\leq \length{\tu x}^c$ for all
$\tu x=(x_1,...,x_p)$.
\end{definition}

The definition of register machines might look rudimentary however, 
the following is easy (but tedious) to prove for any reasonable encoding of integer by Turing machines.

	\begin{theorem}
	A function $f$ from $\N^p\rightarrow \Z$ is computable in polynomial time on Turing machines iff it is computable in polynomial time on register machines.
      \end{theorem}

\subsection{A characterization of polynomial time}

The above result shows that function defined by linear length-ODE from function computable in polynomial time, are indeed polynomial time. We are now ready to introduce a recursion scheme based on solving linear differential equation to capture polynomial time.

%
%
%
%


\begin{remark}
Since the function we define take their values in $\N$ and have output
in $\Z$, composition is an issue. Instead of considering restrictions
of these function with output in $\N$ (which is always possible, even
by syntactically expressible constraints), we simply admit that
composition may not be defined in some cases.
\end{remark}

\begin{definition}
	Let $\derivlength$ be the smallest subset of  functions,
	 that contains   $\mathbf{0}$, $\mathbf{1}$, projections
         $\projection{i}{p}$,  the length function  $\length{x}$,
         the addition function $x \plus y$, the subtraction function $x \minus y$, the multiplication function $x\times y$ (often denoted $x\cdot y$), the sign function $\sign{x}$
	and closed under composition (when defined)  and linear length-ODE
        scheme. \end{definition}


\begin{remark}
	As our results will show, the definition of $\derivlength$ would remain the same by considering closure under any kind of $\lengt$-ODE with $\lengt$ satisfying the hypothesis of Lemma~\ref{lem:fundamental observation linear length ODE}.  
\end{remark}



\begin{example}
A number of natural functions are in $\derivlength$. the following result is immediate by inspection of the example from Section~\ref{sec:programming with ODE} and~\ref{sec:restrict}.
%
	Functions $2^{\length{x}}$, $2^{\length{x}\cdot \length{y}}$, $\cond{x}{y}{z}$, $\suffix(x,y)$, 
	$\lfloor \sqrt{x}\rfloor $, $\lfloor \frac{x}{y}\rfloor$, $2^{\lfloor \sqrt{x}\rfloor}$
	all belong to $\derivlength$.
\end{example}

\begin{theorem}\label{th:ptime characterization 2}
	$\derivlength= \FPtime$
      \end{theorem}

\begin{proof}
The inclusion $\derivlength \subseteq \FPtime$ is a consequence of the fundamental observation proved in Lemma~\ref{lem:fundamental observation linear length ODE}, on the fact that arithmetic operations that are allowed can be computed in polynomial time and that $\FPtime$ is closed under composition of functions.

We now prove that  $\FPtime\subseteq \derivlength$.
Let $f:\N^p\longrightarrow \N$ be computable in polynomial time and $M$ a $k$ registers machine that compute $f$ in time $\length{\tu x}^c$ for some $c\in \N$.
We first describe the computation of $M$ by simultaneous recursion scheme on length for functions  $R_0(t,\tu x), ..., R_k(t,\tu x)$  and $\inst(t,\tu x)$ that give, respectively, the values of each register and the label of the current instruction at time $\length{t}$.

We start with an informal description of the characterization. Initializations of the functions are given by:
$R_0(0,\tu x)=0, R_1(0,\tu x)=x_1$, \dots, $R_p(0,\tu x)=x_p$, $R_{p+1}(0,\tu x)=\cdots = R_k(0,\tu x)=0$ et $\inst(0,\tu x)=0$.
Let $m\in \N$ be the number of instructions of $M$ and let $l\leq m$. Recall that, for a function $f$, $\dderiv{f}{L}(t,\tu x)$ represents a manner to describe $f(t+1,\tu x)$ from $f(t,\tu x)$ when $L(t+1)=L(t)+1$.  We denote by, $\nextI_l^{I}$, $\nextI_l^{h}$, $h\leq k$, the evolution of the  instruction function and of register $R_h$ after applying instruction $l$ at any such instant $t$. They are defined as follows:

\begin{itemize}
\item If instruction of label $l$ if of the type $R_j:=R_j+R_k$, then:

\begin{itemize}
	\item $\nextI_l^{I}=1$ since $\inst(t+1,\tu x)=\inst(t,\tu x)+1$
	\item $\nextI_l^{j}=R_k(t,\tu x)$ since $R_j(t+1,\tu x)=R_j(t,\tu x)+R_k(t,\tu x) $
	\item $\nextI_l^{h}=0$ since $R_h(t,\tu x)$ does not change for  $h\neq j$
\end{itemize}

\item If instruction of label $l$ if of the type $R_j:=R_j-R_k$, then:

\begin{itemize}
	\item $\nextI_l^{I}=1$ since $\inst(t+1,\tu x)=\inst(t,\tu x)+1$
	\item $\nextI_l^{j}=-R_k(t,\tu x)$ since $R_j(t+1,\tu x)=R_j(t,\tu x) -R_k(t,\tu x)  $
	\item $\nextI_l^{h}=0$ since $R_h(t,\tu x)$ does not change for  $h\neq j$
\end{itemize}

\item If instruction of label $l$ if of the type $R_j:=\ell$, for
  $\ell \in \{0,1\}$ then:

\begin{itemize}
	\item $\nextI_l^{I}=1$ since $\inst(t+1,\tu x)=\inst(t,\tu x)+1$
	\item $\nextI_l^{j}=\ell-R_j(t,\tu x)$ since $R_j(t+1,\tu x)=\ell  $
	\item $\nextI_l^{h}=0$ since $R_h(t,\tu x)$ does not change for  $h\neq j$
\end{itemize}

\item If instruction of label $l$ if of the type $\mathsf{if}$  $R_j = 0 \ \textsf{goto} \ p$, then:

\begin{itemize}
	\item $\nextI_l^{I}=\cond{R_j(t,\tu x)}{p-\inst(t,\tu x)}{1}$ since, in case $R_j(t,\tu x)=0$ instruction number goes from $\inst(t,\tu x)$ to $p$.
	\item $\nextI_l^{h}=0$
\end{itemize}

\item If instruction of label $l$ if of the type $\mathbf{Halt}$, then:

\begin{itemize}
	\item $\nextI_l^{I}=0$ since the machine stays in the same instruction when halting
	\item $\nextI_l^{h}=0$.
\end{itemize}
\end{itemize}

The definition of function $\inst$ by derivation on length is now given by (we use a more readable "by case" presentation):

\[
\dderivl{\inst}(t,\tu x)= \case
\left\{\begin{array}{l}
\inst(t,\tu x)=1 \quad \nextI_1^I \\
\inst(t,\tu x)=2 \quad \nextI_2^I\\
\vdots \\
\inst(t,\tu x)=m \quad \nextI_m^I\\
\end{array}
\right.
\]

Expanded as an arithmetic expression, this give: 

\[
\dderivl{\inst}
(t,\tu x)= \sum_{l=0}^m \big(\prod_{i=0}^{l-1} \sign{\inst(t,\tu x)- i}\big)\cdot \signcomp{\inst(t,\tu x) - l}\cdot \nextI_l^I
\]

Note that each $\nextI_l^I$ is an expression in terms of $\inst(t,\tu x)$ and, in some cases, in $\sign{R_j(t,\tu x)}$, too (for a conditional statement).
Similarly, for each $j\leq k$:

\[
\dderivl{R_j}
(t,\tu x)= \sum_{l=0}^m \big(\prod_{i=0}^{l-1} \sign{\inst(t,\tu x)- i}\big)\cdot \signcomp{\inst(t,\tu x)- l}\cdot \nextI_l^j
\]

It is easily seen that, in each of these expressions above, there is
at most one occurence of $\inst(t,\tu x)$ and $R_j(t,\tu x)$ that is
not under the scope of an essentially constant function (i.e. the sign
functions). Hence, the expressions are of the prescribed form.


We know  $M$ works in time $\length{\tu x}^c$ for some fixed
$c\in\N$. Both functions $\length{\tu x}=\length{x_1}+ \ldots
\length{x_p}$ and $B(\tu x)=2^{\length{\tu x}\cdot \length{\tu x}}$ are in $\derivlength$. It is easily seen that : $\length{\tu x}^c\leq B^{(c)}(\length{\tu x}))$ where $B^{(c)}$ is the $c$-fold composition of function $B$.

We can conclude by setting $f(\tu x)=R_0(B^{(c)}(\max(\tu x)),\tu x)$.
\end{proof}

The following normal form theorem can also be obtained (Compared to
Definition~\ref{SPLode}, no function $\tu h$ is allowed on the right
hand side):

%
%


\begin{definition}[Normal linear $\lengt$-ODE (N$\lengt$-ODE)]~\label{def:system of SLL
    ODE} Functions $\tu f 
  $ are definable by a normal linear  $\lengt$-ODE if it corresponds to the
  solution of $\lengt$-ODE
%
$\dderivL{\tu f(x,\tu y)}=   \tu u(\tu f(x,\tu y), x,\tu y)$
%
\noindent where $\tu u$ is \textit{essentially linear} in $\tu f(x, \tu y)$.

\end{definition}





\begin{definition}[$\sll$]
	A function $\tu f:\N^k\to \N^k$ is in $\sll$ if there exists
        $\tu g: \N^{k+1}\to \N$ and $h:\N^k\to \N$ such that:
	\begin{itemize}
		\item $g$ is solution of a normal linear length-ODE 
			$
	\dderiv{\tu g(x,\tu y)}{\length{x}} =   \tu u(\tu g(x,\tu y), x,\tu y);
		$

		\item $h$ is the solution of a single linear length-ODE;

		\item 	and, for all $\tu y\in \N^k$:
%
for some integer $c$. 

	\end{itemize}

\end{definition}

From the proof of Theorem~\ref{th:ptime characterization 2} the result below can be easily obtained. It expresses that composition need to be used only once as exemplified in the above definition.

\begin{theorem} \label{th:pspace}
	$\sll= \FPtime$
      \end{theorem}

\begin{proof}
	In the proof of Theorem~\ref{th:ptime characterization 2}, the
        definition of each function $\inst$, $R_0$,..., $R_k$ are done
        through a linear system of $\sll$-ODE that uses only the basic
        arithmetic and sign functions. Composition is used only to
        bound the computation by $B^{(c)}(\max(\tu x))$, whose
        definition can be obtained through a simple length-ODE.
      \end{proof}

      \section{Further works}
      \label{sec:extension}

 Previous ideas can be extended to provide a characterization of
 $\FPspace$ by considering random access machines (RAM) instead of register
      machines (see Appendix \ref{sec:ram} for definitions) with
      specific instructions sets. Depending on the set of basic
      operations allowed in the RAM model, polynomial time computation
      relates to very different complexity classes as witnessed by the
      following statements (see formal statement and proof of Theorem \ref{theorem:ram for
        pspace2} in appendix): 
%
	\begin{enumerate}
		\item A function $f:\N^k\to \Z$ is computable in polynomial time, i.e. is in $\FPtime$, iff it is computable in polynomial time on a $\{+,-\}$-RAM with unit cost.
		\item A function $f:\N^k\to \Z$ is computable in $\FPspace$ iff it is computable in polynomial time on a $\{+,-, \times, \div\}$-RAM with unit cost.
	\end{enumerate}

Second item follows from the following arguments:  It as been
proved in \cite{Galota:2005bo}, that a function $f$ is in $\FPspace$
iff it is the difference of two functions $f_1,f_2:\N\to \N$ in
$\cPspace$, the class of functions that counts the number of accepting
computations of a non deterministic polynomial space Turing machine.
It follows from the result from \cite{Bertoni:1981fs}, that a
function is computable in polynomial time on a
$\{+,\moins, \times, \div\}$-RAM if and only if it belongs to
$\cPspace$. 

Using random access machines (RAM) instead of register
      machines,  $\FPspace$ can be shown to correspond to functions of
      type $f(\tu y) = g_1(h(\tu y), \tu y)$ where $\tu g$ is defined
      as a specific class of polynomial length-ODE with substitutions,
      and conversely. 

      We leave this characterization for future work, as we believe
      that this statement can be improved to an even simpler
      statement, and scheme.

\newpage

\appendix

\section{Discrete Calculus}
\label{sec:dcal}

The following text is based on
\cite{gleich2005finite,izadi2009discrete,urldiscretecalculuslau}: We
do so using intentionally some notations from continuous ODEs in order
to help understanding to people familiar to classical continuous
theory.  We provide proofs for most of the statements, but some of the
proofs are not repeated here, as they just follow from easy
computations, or as they are classical and can be found in these
references.



Discrete ODEs are basically usually intended to concern functions over the
integers of type $\tu f: \N^p \to \Z^q$, but its statements and concepts
considered in this section are also valid more generally for functions
of type $\tu f: \Z^p \to \Z^q$, for some integers $p,q$, or even functions
$\tu f: \R^p \to \R^q$.

The basic idea is to consider the following concept of
derivative.

\begin{definition}[Discrete Derivative] The discrete derivative of
	$\tu f(x)$ is defined as $\Delta \tu  f(x)= \tu f(x+1)-\tu
        f(x)$. We will also write in this article $\tu f^\prime$ for $\Delta f(x)$ to help to understand
	statements with respect to their classical continuous counterparts. 
	
\end{definition}

\begin{remark}
	The previous concept corresponds to the right derivative. We can 
	write $\tu f^\prime_r(x)$ or $\Delta^+ \tu f(x)$ to emphasise that fact.
	A left derivative version could also be considered: This would
	corresponds to $\Delta^- \tu f(x)= \tu f(x-1)-\tu f(x)$, that we will sometimes
	write $\tu f^\prime_l(x)$.

	In the rest of this section, we will only talk about above right
	derivative, but all results could easily be adapted to deal with left
	derivative.  Actually, left and right derivatives are related by the
	following observation:
	
	\begin{lemma}[Left vs Right Derivative] \label{refbas}
		When $x, x+1, x-1$ fall in the domain of $\tu f$, we always have:
		\begin{eqnarray*}
		\tu 	f^\prime_l(x+1) &=& - \tu  f^\prime_r(x) \\
		\tu	f^\prime_l(x)  &=& - \tu f^\prime_r(x-1) \\
		\end{eqnarray*}
	\end{lemma}
\end{remark}







\subsection{Some basic statements}

\begin{theorem}[Linearity]
	For any functions $\tu f$ and $\tu g$, and constant $\tu c$:
	\begin{eqnarray*}
		(\tu f(x)+\tu g(x))^\prime  &=&  \tu f^\prime (x) +
                                                \tu g^\prime(x) \\
		\ (\tu c \cdot \tu f(x))^\prime  &=& \tu c  \cdot \tu f^\prime (x) \\
	\end{eqnarray*}
\end{theorem}

\begin{theorem}[Inverse]
	Consider $x \neq 0$.
	$$\left(\frac{1}{x}\right)^\prime= - \frac{1}{x} \cdot \frac{1}{x+1}$$
\end{theorem}

\begin{theorem}[Division]
	Assume $g(x) \neq 0$, and $g(x+1)\neq 0$.
	$$\left(\frac{f}{g} \right)^\prime=
	\frac{f^\prime(x) g(x) -f(x)g^\prime(x) }{g(x)g(x+1)}$$
\end{theorem}

\begin{theorem}[Product] \label{th:prod}
	\begin{eqnarray*}
		(\tu f(x) \tu g(x))^\prime &=& \tu f^\prime(x) \tu
                                               g(x+1)+\tu f(x) \tu g^\prime(x) \\
		&=& \tu f(x+1) \tu g^\prime(x) + \tu f^\prime(x) \tu g(x) \\
	\end{eqnarray*}
	
\end{theorem}

\subsection{The discrete integral and primitive}

\begin{definition}[Discrete Integral]
	Given some function $\tu f(x)$, we write $\dint{a}{b}{\tu f(x)}{x}$
	as a synonym for
	\begin{itemize}
		\item $$\dint{a}{b}{\tu f(x)}{x}=\sum_{x=a}^{x=b-1}
                  \tu f(x)$$
		when $a<b$,
		(pay attention to the fact that the bound is $b-1$ on right, and $b$
		on left)
		\item
		$0$ when $a=b$,
		\item
		and, when $a>b$:
		$$\dint{a}{b}{\tu f(x)}{x}=- \dint{b}{a}{\tu f(x)}{x}$$
		
	\end{itemize}
	
\end{definition}

The following holds from a basic computation (from 
telescope formula):

\begin{theorem}[Fundamental Theorem of Finite Calculus]
	Let $\tu F(x)$ be some function.
	Then,
	$$\dint{a}{b}{\tu F^\prime(x)}{x}= \tu F(b)-\tu F(a).$$
\end{theorem}

As a consequence:

\begin{definition}[Discrete Primitive]
	Let $\tu f(x)$ be some function, and $\tu C$ some constant (of
        suitable dimension if $\tu f$ is vectorial).
	Then the function
	
	$$\tu F(x)= \tu C+ \sum_{x=0}^{x-1} \tu f(x)$$
	
	\noindent is such that $\tu F^\prime(x)=\tu f(x)$
	and $\tu F(0)=\tu C$. As expected,
	$\tu F$ is called a primitive of $\tu f(x)$.
      \end{definition}


\begin{corollary}
	Let $\tu f(x)$ be some function, and $\tu F(x)$ its primitive.
	
	$$\tu F(b)-\tu F(a)= \dint{a}{b}{\tu F^\prime(x)}{x} = \sum_{x=a}^{x=b-1} \tu f(x)$$
	
	And:
	
	$$\sum_{x=a}^{x=b} \tu f(x) = \tu F(b+1)- \tu F(a)$$
      \end{corollary}

 \begin{remark}
   Recall that for classical continuous
 derivative if $F(x) = \int_{a(x)}^{b(x)} f(x,t) dt$ then
 $$F'(x) = f(x,b(x)) b^\prime(x) - f(x,a(x)) a^\prime (x) +
 \int_{a(x)}^{b(x)} \frac{\partial f}{\partial x} (x,t) dt$$

 This generalizes to the following:
\end{remark}

\ENONCEINTEGRALPARAMETER{\label{derivintegral}}

 \begin{proof}
\begin{eqnarray*}
\tu F(x+1) - \tu F(x) &=& \sum_{t= a(x+1)}^{b(x+1) - 1} \tu f(x+1,t) -
                          \sum_{t= a(x)}^{b(x) - 1} \tu f(x,t)   \\
  &=& \sum_{t= a(x)}^{b(x)-1} \left( \tu  f(x+1,t) - \tu f(x,t)  \right)  
   +   \sum_{t=a(x+1)} ^{t= a(x)-1} \tu f(x+1,t) + \sum_{t=b(x)}
      ^{b(x+1)-1} f(x+1,t) \\
  &=& \sum_{t= a(x)}^{b(x)-1}  \frac{\partial \tu f}{\partial
      x} (x,t)   + \sum_{t=a(x+1)} ^{t= a(x)-1} \tu f(x+1,t) + \sum_{t=b(x)}
      ^{b(x+1)-1} \tu f(x+1,t) \\
  &=& \sum_{t= a(x)}^{b(x)-1}   \frac{\partial \tu f}{\partial
      x} (x,t)  + \sum^{t=-a(x+1)+a(x)-1} _{t=0} \tu f(x+1,a(x+1)+t)\\
  && + \sum_{t=0}
      ^{b(x+1)-b(x)-1} \tu  f(x+1,b(x)+t) \\
\end{eqnarray*}
\end{proof}

\subsection{Integration by part}

\begin{theorem}[Integration by part]
	$$\dint{a}{b}{\tu u(x) \tu v^\prime(x)}{x} = [\tu u(x)\tu
        v(x)]_a^b - \dint{a}{b}{\tu u^\prime(x) \tu v(x+1)}{x}
		$$
	
	where $[\tu u(x)\tu v(x)]_a^b$ stands for $\tu u(b)\tu v(b)-\tu u(a)\tu v(a)$
\end{theorem}

\begin{proof}
	Write $(\tu u(x)\tu v(x))^\prime=\tu u(x)\tu v^\prime(x)+\tu u^\prime(x)\tu v(x+1)$, and hence
	$\tu u(x)\tu v^\prime(x)= (\tu u(x)\tu v(x))^\prime - \tu
        u^\prime(x)\tu v(x+1)$. Then integrate. 
\end{proof}

\subsection{Derivative of a composition}

The following can be established:

\begin{theorem}[Derivative of $\tu f \circ g$]
	$$\tu f(g(x))^\prime= \dint{0}{ g^\prime(x)}{\tu f^\prime(g(x)+k)}{k}$$
\end{theorem}

\begin{proof}
	Write
	\begin{oureqnarray}
	\tu f(g(x+1))-\tu f(g(x)) &=& \dint{g(x)}{g(x+1)}{\tu f^\prime(t)}{t} \\
	&=& \dint{0}{g(x+1)-g(x)} {\tu f^\prime(g(x)+k)}{k} \\
	&=& \dint{0}{g^\prime(x)} {\tu f^\prime(g(x)+k)}{k}
	\end{oureqnarray}
\end{proof}

\subsection{Falling power}

With analogy with the concept of derivative of a power, this is
traditional to
define ($m$ stands for some natural integer).

\begin{definition}[Falling power] The expression $x$ to the $m$
	falling is denoted by $x^{\underline{m}}$ (sometimes denoted by
	$(x)_m$) stands for  $$x^{\underline{m}}=x\cdot (x-1)\cdot (x-2)\cdots(x-(m-1)).$$
\end{definition}

This is motivated by the following observation:

\begin{theorem}[Derivative of a falling power] The discrete derivative
	of a falling power having exponent $m$ is $m$ times the next lowest
	falling power: That is
	$$ (x^{\underline{m}})^\prime  = m \cdot x^{\underline{m-1}}$$
\end{theorem}

\subsection{Exponential}

\begin{theorem}[Exponential $c^x$]
	Let $c$ be some positive constant. We have
	$$(c^{x})^\prime = (c-1) \cdot c^{x}.$$
	
	In particular
	$$ (2^{x})^\prime =   2^{x}.$$
\end{theorem}

More generally,

\begin{theorem}[Exponential $c^{f(x})$]
	Let $c$ be some positive constant. We have
	$(c^{f(x)})^\prime = (c^{f^\prime(x)}-1) \cdot c^{f(x)}.$
\end{theorem}

\subsection{Falling Exponential}

In a spirit similar to the falling power above, we propose to
introduce the following concept. This seems not standard (as far as
the authors know, but this seems to be of clear interest).

We assume $x \in \N$ in the following discussions.



\begin{definition}[Falling exponential]
	Given some function $\tu U(x)$, the expression $\tu U$ to the
	falling exponential $x$,
	denoted by $$\fallingexp{\tu U(x)}= (1+ \tu U^\prime(x-1)) \cdots
        (1+ \tu U^\prime(1)) \cdot (1+ \tu U^\prime(0))  = 
	\prod_{t=0}^{t=x-1} (1+ \tu U^\prime(t)). $$
        	with the convention that $\prod_{0}^{0}=\tu {id}$, where $\tu
        {id}$ is the identity (e.g. $1$ for the scalar case). 
\end{definition}

This is motivated by the following two observations:

\begin{lemma} For all $x \in \Z$,
	$2^x=\fallingexp{x}$
\end{lemma}

\begin{theorem}[Derivative of a falling exponential] The discrete
	derivative of a falling exponential is given by
$$\left(\fallingexp{\tu U(x)}\right )^\prime = \tu U^\prime(x) \cdot
	\fallingexp{\tu U(x)} $$
	
\end{theorem}

In particular, we can easily build towers of exponentials using
polynomial ordinary differential equations (ODEs):

\subsection{Solving some particular ODEs}

We will here consider the discrete variants of some particular
(linear) ODEs. 

\begin{remark}
	We assume implicitly in all this section that $x \in
	\N$, i.e. we discuss solutions of the ODEs over the domain $\N$. 
\end{remark}


\begin{remark}
	Recall that the solution of $\tu f^\prime(x)=\tu b(x)$, $\tu f(0)=0$ for classical continuous derivatives is
	given by $$\tu f(x)  = \int_0^x \tu b(t) dt.$$ Here we have something very
	similar:
\end{remark}

\begin{lemma}[Solution of ODE $\tu f^\prime(x)=\tu b(x)$]
	The solution of $\tu f^\prime(x)=\tu b(x)$, $\tu f(0)=\tu \zero$
	is $$\dint{0}{x}{\tu b(t)}{t}
	$$
\end{lemma}

\begin{proof}
	Consider $\tu f(x)=\dint{0}{x}{\tu b(t)}{t}=
	\sum_{t=a}^{t=x-1} \tu b(t)$.
	For $x=0$, we have $\tu f(x)=0$.
	For $x>0$, we have $\tu f^\prime(x)=\tu f(x+1)-\tu f(x)=\tu b(x)$.
      \end{proof}


\begin{remark}
	Recall that the solution of $ f^\prime(x)=  a(x) f(x)$, $
	f(0)=1$ (respectively: or more generally for the vectorial constant case 
        $\tu f^\prime(x)= \tu A \cdot  \tu f(x)$, $\tu
	f(0)=\tu 1$)  for classical
	continous derivatives is
	given by $$ f(x)  = e^{\int_0^x a(t) dt}$$ (resp. $\tu
        f(x)  = e^{t\tu A} $).  
        Something very
	similar holds in the discrete setting:
\end{remark}

\begin{lemma}[Solution of ODE $\tu f^\prime(x)=\tu A(x) \cdot \tu f(x)$] \label{lem:af}
	The solution of $\tu f^\prime(x)=\tu A(x) \cdot \tu f(x)$ is
	$$\fallingexp{\dint{0}{x}{\tu A(t)}{t}} \cdot \tu f(0).$$
\end{lemma}

Notice that
$$\fallingexp{{\dint{0}{x}{\tu A(t)}{t}}}=(1+\tu A(x-1)) \cdots (1+\tu
A(1)) \cdot (1+\tu A(0)) = \prod_{t=0}^{t=x-1} (1+\tu A(t)).$$

\begin{proof}
	Consider $\tu f(x)= 
	\fallingexp{{\dint{0}{x}{\tu A(t)}{t}}} \cdot \tu f(0) 
	.$
        This values $\tu f(0)$ in $0$. For $x>0$,  we have $\tu f(x+1) =
	(1+\tu A(x)) \cdot \tu f(x)$ and
	hence $\tu f(x+1)- \tu f(x) = \tu A(x) \cdot \tu f(x) $.  
\end{proof}


\subsection{Solving affine ODEs}

We now go to affine (also called linear) ODE $\tu f^\prime(x)=\tu A(x) \cdot \tu f(x) + \tu B(x)$.  This affine
ODEs play a key role in this article. 

\begin{remark}
	The solutions of $ f^\prime(x)=  a(x) \cdot  f(x) +  b(x)$
	for
	classical continous derivatives are
	\begin{equation} \label{eq:affine2}
	f(x) = f(0) e^{\int_0^{x} a(t)
		dt}  + \int_{0}^{x} b(u) e^{\int_u^{x} a(t)
		dt }
	du
	\end{equation}
\end{remark} 

\begin{remark} It is usually obtained by variation of parameter
	method: we search $f(x)$ of the form $f(x)= f_1(x) k(x)$ where
	$f_1(x)$ is solution of $f_1^\prime(x)=a(x) f_1(x)$
        (from above discussions,
	we hence have $f_1(x)= f_1(0) e^{\int_0^{x} a(t) dt}$).  Indeed, the trick
	is then that must have
	$f_1^\prime(x) k(x) + f_1(x) k^\prime(x) = a(x) f_1(x) k(x) + b(x)$: Factors of
	$k(x)$ cancels, and we get $f_1(x) k^\prime(x) = b(x)$.
        
        Multiplying by $e^{-\int_0^{x} a(t) dt}$ both sides, we get
        $ e^{-\int_0^{x} a(t) dt} f_1(x) k^\prime(x) =
        e^{-\int_0^{x} a(t) dt} b(x)$ which simplifies to
        $f_1(0)  k^\prime(x) = e^{-\int_0^{x} a(t) dt}  b(x) $, equation in
        $f_1(0) k(x)$ 
          than can be solved by a simple integral:
          $$ f_1(0) k(x) = f_1(0) k(0) + \int_{0}^{x} b(u)  e^{-\int_0^{u} a(t)
            dt} du,$$
          and then reporting the expression of $f_1(0) k(x)$
          \begin{eqnarray*}
            f(x) &=& f_1(x) k(x) \\
            &=&  f_1(0) e^{\int_0^{x} a(t) dt} k(x) \\
          &=& f_1(0) k(0)   e^{\int_0^{x} a(t) dt}  + e^{\int_0^{x} a(t) dt} \int_{0}^{x} b(u)  e^{-\int_0^{u} a(t)
              dt} du \\
            &=& f_1(0) k(0) e^{\int_0^{x} a(t)
		dt}  + \int_{0}^{x} b(u)  e^{\int_u^{x} a(t)
		dt } 
                du.\\
          \end{eqnarray*}
          	Considering value in $0$, we realize that $f_1(0) k(0)$ is actually $f(0)$ and
	obtain the above solution.
        
\end{remark}

\begin{remark}
	The solution \eqref{eq:affine2} is the sum of a solution to $f^\prime(x)=a(x) f(x)$, i.e. of the ODE
	with the non-linear term, and of a solution that values
	$0$ in $0$.
      \end{remark}


      \begin{remark}
        This extends for the vectorial case for classical continuous
        derivatives. This is usually obtained using
        the concept of resolvant: resolvant $\tu R(\tu x,\tu x_0)$ is by
        definition such that solutions of
        $\tu f^\prime(x)=  \tu A(x) \cdot \tu f(x)$ with $\tu f(0) =\tu y_0$
        correspond to  $\tu f(x) = \tu
        R(x,x_0) \cdot \tu y_0$.

        In the case where $\tu A(x) = \tu A$ is constant, the
        resolvant is given by $\tu
        R(x,x_0) = e^{(x-x_0) \cdot \tu A}$. 

	The solutions of $ \tu f^\prime(x)=  \tu A(x) \cdot  \tu f(x)
        +  \tu B(x)$
	for
	classical continous derivatives are then given by 
	\begin{oureqnarray} \label{eq:affine3}
          \tu f(x) &=& \tu R(x,0) \cdot \tu f(0) + \tu R(x,0)  \cdot \int_{0}^{x} \tu R(0,u) \cdot
                       \tu b(u) du \\
          &=& \tu R(x,0) \cdot \tu f(0) + \int_{0}^{x} \tu R(x,u) \cdot
                       \tu b(u) du 
        \end{oureqnarray}
        in the general case. 

\end{remark}

In the discrete case, something similar holds. It is detailed below in
the context of functions with several variables to be used, as it is,
later.







\begin{lemma}[Solution of ODE $\tu f’(x,\tu y)= \tu A (x,\tu y) \cdot
	\tu f(x,\tu y)
	+ \tu B(x,\tu y)$] \label{def:solutionexplicitedeuxvariablesA}
	For matrices $\tu A$ and vectors $\tu B$ and $\tu G$,
	the solution of equation $\tu f’(x,\tu y)= \tu A(x,\tu y) \cdot \tu f(x,\tu y)
	+  \tu
	B (x,\tu y)$  with initial conditions $\tu f(0,\tu y)= \tu G(\tu y)$ is
	$$\left( \fallingexp{\dint{0}{x}{\tu A(t,\tu y)}{t}} \right) \cdot \tu G (\tu y) +
	\dint{0}{x}{ \left(
		\fallingexp{\dint{u+1}{x}{\tu A(t,\tu y)}{t}} \right) \cdot
              \tu B(u,\tu y)} {u}.$$
	
\end{lemma}

Notice that this can also be written:

$$\left( \prod_{t=0} ^{t=x-1} (1+\tu A (t,\tu y)) \right) \cdot \tu G(\tu y) +
\tu B (x-1,\tu y) + 
\sum_{u=0}^{x-2} \left(
\prod_{t=u+1}^{x-1} (1+\tu A(t,\tu y)) \right) \cdot \tu B (u,\tu y). $$

This can also be expressed by simpler expression:

$$
\sum_{u=-1}^{x-1}  \left(
\prod_{t=u+1}^{x-1} (1+\tu A(t,\tu y)) \right) \cdot  \tu B(u,\tu y).
$$

with the (not so usual) conventions that $\prod_{x}^{x-1} \tu \kappa(x) = 1$ and $\tu
B(-1,\tu y)=\tu G(\tu y)$.







	
	
	
	

Before getting to the proof, we start by another Lemma:

\begin{lemma}[Solution of $\tu H'(x)= \tu A(x) \cdot \tu H(x) + \tu
  B(x)$ with $\tu H(x)=0$] \label{special} 
  The solution of $$\tu H'(x)= \tu A(x) \cdot \tu H(x) + \tu
  B(x)$$ with $\tu H(x)=0$ is given by 

  $$\tu H(x) = 	\dint{0}{x}{ \left(
		\fallingexp{\dint{u+1}{x}{\tu A(t)}{t}} \right) \cdot
              \tu B(u)} {u}.$$
          \end{lemma}

\begin{proof}
  Consider above expression. We then have. 
            \begin{eqnarray*}
              \tu H(0) &=& 0 \mbox { and  from Lemma \ref{derivintegral}} \\ 
              \tu H ^\prime (x) &=& \dint{0}{x}{ \tu A(x) \cdot \left(
		\fallingexp{\dint{u+1}{x}{  \tu A(t)}{t}} \right) \cdot
              \tu B(u)} {u} +
            \fallingexp{\dint{x+1}{x+1}{\tu A(t)}{t}} \cdot \tu B(x)
              \\
                       & = & \tu A(x) \cdot \tu H(x) + \tu B(x). \\
            \end{eqnarray*}


\end{proof} 

We can now go to the proof of Lemma
\ref{def:solutionexplicitedeuxvariables}.

\begin{proof}
  From linearity of derivation, we must have $\tu f_1(x, \tu y) =\tu f(x,\tu y) - \tu
  H(x, \tu y)$  solution of $\tu f^\prime_1(x,\tu y )=\tu A(x,\tu
  y)\cdot \tu f_1(x,\tu y)$, where $\tu H(x, \tu y)$ satisfies $\tu
  H(0,\tu y) = 0$ and $\tu H'(x, \tu y) = \tu A (x, \tu y) \cdot \tu
  H(x, \tu y) + \tu B(x, \tu y)$.  The solution of latter equation is
  given by Lemma \ref{special}.
  
A general solution $f_1$ of former   
equation is (see above) 
%
\[
\tu f_1(x,\tu y) = \fallingexp{ \dint{0}{x}{\tu A(t,\tu y)}{t}} \cdot \tu f_1(0,\tu y). 
\]
This leads to the above expression. 
\end{proof}

\begin{remark}
Fomula \eqref{soluce} can also be expressed by simpler expression:

$$
\sum_{u=-1}^{x-1} \left(
\prod_{t=u+1}^{x-1} (1+\tu A(t,\tu y)) \right) \cdot \tu B(u,\tu y).
$$

with the (not so usual) conventions that $\prod_{x}^{x-1} \tu \kappa(x) = 1$ and $\tu
B(-1,\tu y)=\tu G(\tu y)$. 

\end{remark}

Exactly the same (first) proof shows that the following generalization
holds:

\begin{lemma}[Solution of ODE $\tu f’(x,\tu y)= \tu A (\tu f(x,\tu y), x,\tu y) \cdot
	\tu f(x,\tu y)
	+ \tu B(f(x,\tu y), x,\tu y)$] \label{def:solutionexplicitedeuxvariablesg}
	For matrices $\tu A$ and vectors $\tu B$ and $\tu G$,
	the solution of equation $\tu f’(x,\tu y)= \tu A (f(x,\tu y), x,\tu y) \cdot
	\tu f(x,\tu y)
	+ \tu B(f(x,\tu y), x,\tu y)$  with initial conditions $\tu f(0,\tu y)= \tu G(\tu y)$ satisfies
	$$\tu f(x,\tu y) = \left( \fallingexp{\dint{0}{x}{\tu A(\tu
              f(t,\tu y) , t,\tu y)}{t}} \right) \cdot \tu G (\tu y) +
	\dint{0}{x}{ \left(
		\fallingexp{\dint{u+1}{x}{\tu A(\tu f(t,\tu y) , t,\tu
                    y)}{t}} \right) \cdot \tu B(\tu f(u,\tu y) ,u,\tu y)}{u}.$$
	
\end{lemma}

In an analog way to above, this can also be written:

\begin{eqnarray}
  \tu f(x, \tu y) &=& \left( \prod_{t=0} ^{t=x-1} (1+\tu A (\tu f(t,\tu y) , t,\tu
  y)) \right) \cdot \tu G(\tu y) + \tu B (x-1, \tu y) \\
&&+ 
\sum_{u=0}^{x-2} \left(
\prod_{t=u+1}^{x-1} (1+\tu A(\tu f(t,\tu y) , t,\tu y)) \right) \cdot
\tu   B (\tu f(u,\tu y) , u,\tu y). 
\end{eqnarray}

or as

$$\tu f(x, \tu y) =
\sum_{u=-1}^{x-1 } \left(
\prod_{t=u+1}^{x-1} (1+\tu A(f(t,\tu y), t,\tu y)) \right) \cdot \tu B(f(u,\tu y) , u,\tu y).
$$

with the conventions that $\prod_{x}^{x-1} \tu \kappa(x) = 1$ and $\tu
B(\cdot , -1,\tu y)=\tu G(\tu y)$.




	
	



\subsection{Derivative of some particular functions}

We now provide some other examples of functions with their derivative.

\begin{theorem}[$\sin$, $\cos$]
	We have:
	\begin{eqnarray*}
		\sin(x)^\prime &=& 2 \cdot \sin\left(\onehalf\right) \cdot \cos \left(x+\onehalf\right) \\
		\cos(x)^\prime &=& -2 \cdot  \sin\left(\onehalf\right) \cdot  \sin \left(x+\onehalf\right) \\
	\end{eqnarray*}
\end{theorem}

\begin{theorem}[$\tan$]
	Whenever $cos(x) \neq 0$ and $\cos(x+1)\neq 0$, we have:
	\begin{eqnarray*}
		\label{eq:tan}
		\tan(x)^\prime&=&\onehalf\sin\left(\onehalf\right)\cos\left(\onehalf\right) \frac{1}{\cos(x)\cos(x+1)}\\
		&=&  \tan(1) \cdot 
		(1+\tan(x) \tan(x+1))
	\end{eqnarray*}
\end{theorem}

%
%
%

\newpage
\section{Random access machines}
      \label{sec:ram}

In the following we consider the random access machine with unit cost as computation model.
Let $Op$ be a set of arithmetic operations. A  $Op$-RAM  is the
collection of a \textbf{potentially infinite} set of registers $(R_i)$ where
${i\in \N^*}$ and two special registers $A,B$.   A program is a finite sequence of ordered labeled instructions, $I_0,...,I_r$ acting
on registers of one of the following type:

\begin{enumerate}

	\item $A:=l$, $B:=l$, $l\in \N$

	\item $A:=A\odot  B$ or $B:=A \odot B$, for $\odot \in Op$

	\item $B:=A$, $A:=B$

	\item $A:=R_A$  meaning that $A$ receive the content of the register whose address is in $A$, provided $A$ is non negative (indirect addressing).

	\item $R_A:=B$ meaning that the register whose address is in (non negative) $A$ receives the content of register $B$.

	\item If $A=B$ then goto $I_i$ else goto $I_j$ (pour tout $i,j\leq r$).

	\item \textsf{halt}
\end{enumerate}

	\begin{definition} Let  $t:\N \rightarrow \N$.
	A function $f:\N^p\rightarrow \Z$ is computable in time  $t$
	by a RAM machine $M$  if:
	\begin{itemize}
		\item when starting in initial configuration with registers $R_1,\dots, R_p$ set to $x_1,\dots,x_{p}$ and  all other registers to $0$ and
		\item starting on the first instruction (of label $0$),
	\end{itemize}

	\noindent machine $M$ ends its computation after at most $t(\length{x})$ instructions where $\length{x}=\length{x_1}+\cdots+\length{x_p}$ and with register $A$ containing $f(x_1,\dots,x_p)$.

	A function is computable in polynomial time by $M$ if there exists $c\in \N$ such that $t(\length{x})\leq \length{x}^c$ for all $x=(x_1,...,x_p)$.
\end{definition}

Depending on the set of basic operations allowed in the RAM model, polynomial time computation relates to very different complexity classes as witnessed by the following result.

\begin{theorem}\label{theorem:ram for pspace2}
	\begin{enumerate}
		\item A function $f:\N^k\to \Z$ is computable in polynomial time, i.e. is in $\FPtime$, iff it is computable in polynomial time on a $\{+,-\}$-RAM with unit cost.
		\item A function $f:\N^k\to \Z$ is computable in $\FPspace$ iff it is computable in polynomial time on a $\{+,-, \times, \div\}$-RAM with unit cost.
	\end{enumerate}
\end{theorem}

\begin{proof}
	Let $f:\N^k\to \Z$, let $f_1,f_2:\N^k\to \N$ defined by   $f_1=\max \{0,f\}$, $f_1=\max \{0,-f\}$. Remark that $f=f_1-f_2$.
The following is easily seen, through a reasonable representation of integers (see, e.g., Definition~\ref{def:representation of integers}) and a straightforward simulation of arithmetic operations: function   $f$ is computable in polynomial time on a  $\{+,-\}$-RAM (resp. $\{+,-, \times, \div\}$-RAM) if and only if $f_1$ and $f_2$ are computable  in polynomial time on a  $\{+,\moins\}$-RAM (resp. $\{+,\moins, \times, \div\}$-RAM).
From that, the first item follows easily by classical simulation between machine models~\cite{vanLeeuwen:1991:HTC}.

It as been proved in \cite{Galota:2005bo}, that  a function $f$ is in $\FPspace$ iff it is the difference of two functions $f_1,f_2:\N\to \N$ in $\cPspace$, the class of
functions that counts the number of accepting computations of a non
deterministic polynomial space Turing machine.  We conclude using the early remarks of the proof  and the result  from  \cite{Bertoni:1981fs}, that a function is computable in polynomial time on a $\{+,\moins, \times,
\div\}$-RAM  if and only if it belongs to $\cPspace$. 
\end{proof}

\end{document}